\documentclass[10pt,journal,compsoc]{IEEEtran}

\usepackage[switch]{lineno}
\usepackage{todonotes}
\usepackage{xcolor}
\definecolor{bluegreen}{RGB}{0,153,153}
\definecolor{bluepurple}{RGB}{102,0,204}
\definecolor{bluered}{RGB}{204,0,102}

\ifCLASSOPTIONcompsoc
  \usepackage[nocompress]{cite}
\else
  \usepackage{cite}
\fi

\usepackage{graphicx}
\usepackage{caption}
\usepackage{amsfonts}
\usepackage{pifont}
\usepackage{amsthm}
\usepackage{multirow}
\usepackage{adjustbox}
\usepackage{colortbl}

\ifCLASSINFOpdf
\else
\fi
\usepackage{amsmath}
\usepackage{algorithmicx}
\usepackage{algorithm}  
\usepackage{algpseudocode}

\usepackage{array}

\ifCLASSOPTIONcompsoc
 \usepackage[caption=false,font=footnotesize,labelfont=sf,textfont=sf]{subfig}
\else
 \usepackage[caption=false,font=footnotesize]{subfig}
\fi
\newtheorem{theorem}{Theorem}

\begin{document}

\sloppy

%
\title{Online Deployment Algorithms for Microservice Systems with Complex Dependencies}
%
%
%
%

\author{Xiang~He,
        Zhiying~Tu,~\IEEEmembership{Member,~IEEE,}
        Markus~Wagner,~\IEEEmembership{Member,~IEEE,}
        Xiaofei~Xu,~\IEEEmembership{Member,~IEEE,}
        and~Zhongjie~Wang,~\IEEEmembership{Member,~IEEE,}%

\IEEEcompsocitemizethanks{
\IEEEcompsocthanksitem X. He, Z. Tu, X. Xu, and Z. Wang are with the Faculty of Computing, Harbin Institute of Technology, China.\protect\\
E-mail: \{hexiang, tzy\_hit, xiaofei, rainy\}@hit.edu.cn
\IEEEcompsocthanksitem M. Wagner is with the School of Computer Science, The University of Adelaide, Australia.\protect\\
E-mail: markus.wagner@adelaide.edu.au
}
}

\markboth{IEEE Transactions on Cloud Computing}%
{He \MakeLowercase{\textit{et al.}}: Online Deployment Algorithms for MicroserviceSystems with Complex Dependencies}
%



\IEEEtitleabstractindextext{%
\begin{abstract}

Cloud and edge computing have been widely adopted in many application scenarios. With the increasing demand of fast iteration and complexity of business logic, it is challenging to achieve rapid development and continuous delivery in such highly distributed cloud and edge computing environment. At present, microservice-based architecture has been the dominant deployment style, and a microservice system has to evolve agilely to offer stable Quality of Service (QoS) in the situation where user requirement changes frequently.
Many research have been conducted to optimally re-deploy microservices to adapt to changing requirements.
Nevertheless, complex dependencies between microservices and the existence of multiple instances of one single microservice in a microservice system have not been fully considered in existing works. This paper defines SPPMS, the Service Placement Problem in Microservice Systems that feature \textit{complex dependencies} and \textit{multiple instances}, as a Fractional Polynomial Problem (FPP) . Considering the high computation complexity of FPP, it is then transformed into a Quadratic Sum-of-Ratios Fractional Problem (QSRFP) which is further solved by the proposed greedy-based algorithms. Experiments demonstrate that our models and algorithms outperform existing approaches in both quality and computation speed. 


\end{abstract}

}

\maketitle

\IEEEdisplaynontitleabstractindextext

%
\IEEEpeerreviewmaketitle

\IEEEraisesectionheading{\section{Introduction}\label{sec:introduction}}

%
%
%
%


 

Cloud computing is widely adopted in scenarios that involve Internet of Things (IoT) devices and mobile devices. It plays an essential role in improving performance and in reducing the energy consumption of local devices by offloading computing tasks from local devices to clouds~\cite{MEC_offloading_challenges}. It also empowers the system to process large amounts of data when the number of devices grows massively~\cite{Big_IoT_Data_cloud}. Furthermore, many other technologies like Edge Computing~\cite{edge_computing} and Fog Computing~\cite{fog_computing} are adopted to further enhance the Quality of Service (QoS) in different scenarios.

However, in cloud/edge/fog computing with monolithic applications, the increasingly complex business logic and variety of user requirements leads to a decline in DevOps performance. Thus, microservice and container technologies have been widely adopted to facilitate service management~\cite{docker_edge, osmotic_computing}. In addition, they ensure continuous delivery and deployment by dividing monolithic applications into several independent microservices and deploying them in containers~\cite{MicroDevOps}. Although the microservice design pattern constitutes an approach to software and systems architecture based on the concept of modularization and boundaries~\cite{Microservices}, the requests between microservices make dependencies between services more complex~\cite{sdg_dependency,version_based}. Moreover, the microservice-based applications in microservice systems share the same microservice set for similar functions, which makes the dependencies between services even more complex.

\figurename~\ref{fig:scheme_example} shows an example of services with complex dependencies, where different applications are highlighted using different colors. Because microservices can have many Application Programming Interfaces (API) with different functions, they could be called by both other services and users, like service 10 and 2 in \figurename~\ref{fig:scheme_example}. Services can also request different services according to different user requirements. For example, service 10 might call service 8 if service 10 is called by users directly, and it might call service 11 instead when called by service 7.


\begin{figure}[!t]
    \centering
    \includegraphics[width=\linewidth]{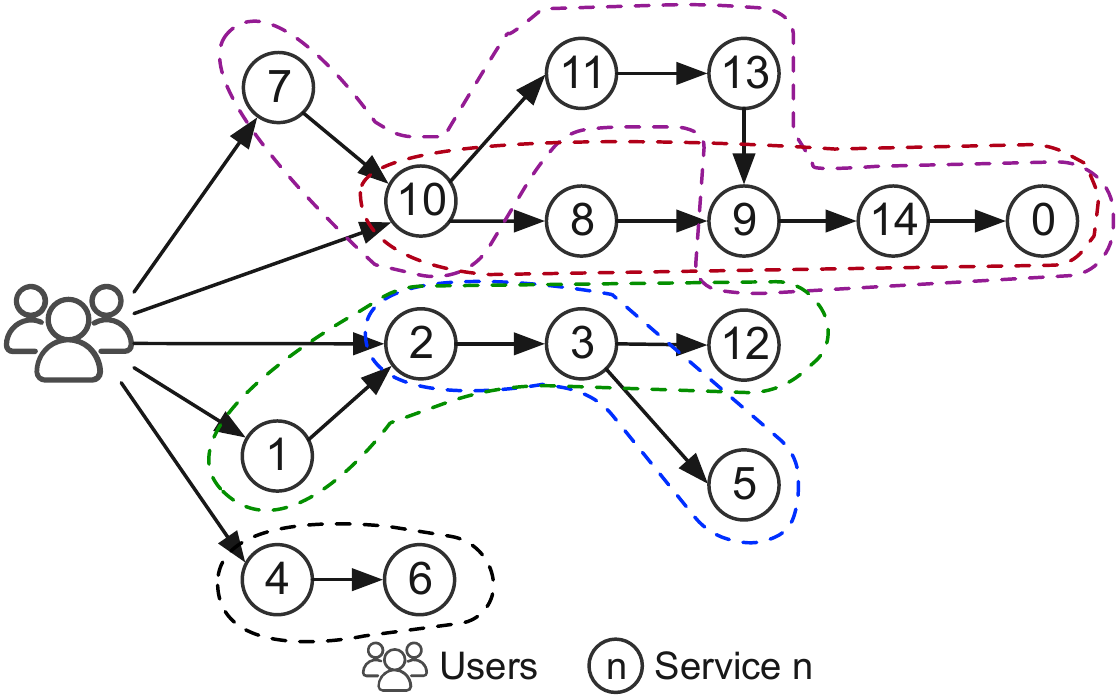}
    \caption{An example of services with complex dependencies}
    \label{fig:scheme_example}
\end{figure}

In operation, the scale of the microservice system can increases as users and services join, and user requirements also change frequently with time. The changes in user requirements can lead to a decline in QoS without evolving the system on time. The evolution represents the change of a system from a previous state to a new state by different technologies. Because re-deploying services according to new requirements at runtime is a common method of the evolution, an efficient online service deployment algorithm is needed to keep the QoS stable provided by the microservice system. The algorithm aims to generate an optimized service deployment scheme based on the current system status, and the system can execute the scheme as the evolution. While there already is some research focussing on the multi-component application placement algorithm for providing the best execution cost or QoS with resource constraint~\cite{TMC_deployment,GMCAPP,LDSPP}, some of the primary features in a microservice system are not fully considered:

\begin{itemize}
    \item \textbf{Multiple Applications with complex dependencies}: A microservice system typically consists of multiple microservice-based applications with complex dependencies, while most studies only consider one application. With multiple applications, the system has different categories of user requirements, while the users only request one function in the system with one application.
    \item \textbf{Multiple Instances}: Multiple instances of each service can be deployed on different servers for load balancing in the microservice system, while some research only assign each service to one server.
\end{itemize}

It is essential to provide the best QoS by deploying service instances to servers considering the complex dependencies with support for multiple applications and instances in microservice systems. It should be noted that the server could be a physical one or a virtual machine. Moreover, the algorithm execution time should at least not exceed the container boot time. Considering the average boot time is 1.5s in~\cite{docker_performance}, although the average boot time depends on the hardware, the algorithm execution time should not exceed 5~10s in order to avoid degrading the user experience.
We call the problem of deploying services in such a system for QoS the Service Placement Problem in Microservice Systems (SPPMS). This paper selects the average response time as the main optimization goal because it plays a vital role in cloud and edge computing.

Our main contributions are as follows:

\begin{itemize}
    \item We define the Service Placement Problem in Microservice Systems (SPPMS) considering the complex service dependencies with support for multiple applications and instances. We formulate the SPPMS as a Fractional Polynomial Problem (FPP) without any restrictions on the topology of servers and services.
    \item We further prove a theorem that helps to rewrite the problem formulation to a Quadratic Sum-of-Ratios Fractional Problem (QSRFP), which can significantly reduce computational complexity compared to FPP. Two efficient greedy algorithms are designed to solve QSRFP based on the theorem, and the optimal algorithm based on the lower bound function and binary search is also introduced for reference.
    \item Experiments are conducted to evaluate the performance of the proposed algorithms in different situations against existing algorithms~\cite{GMCAPP,LDSPP}.
    The results show that our algorithms outperform existing approaches in both quality and speed. Discussions about the selection of proposed algorithms and methods to improve the performance are also presented.  
\end{itemize}

The rest of the paper is organized as follows. Section~\ref{sec:related_work} presents the related work on the deployment problems in cloud and edge computing. Section~\ref{sec:problem_formulation} introduces SPPMS and formulates it as an FPP. Section~\ref{sec:algorithm} shows the conversion from FPP to QSRFP and details the optimal algorithm and our proposed algorithms. Section~\ref{sec:experiments} presents the experiments and the analysis. Section~\ref{sec:conclusion} concludes the paper and outlines possible future work.

\section{Related Work}\label{sec:related_work}

There have already been studies on the service placement problem in cloud/edge/fog computing, and they mainly focus on two different scenarios: the service placement problem without and with dependencies between services. The former assumes the services in the system are independent of each other, and there is no request between them; the latter investigates the placement problem of the multi-component applications considering the dependencies between components.

For the service placement problem without dependencies, works~\cite{QoS_Fog_placement}, \cite{follow_me_at_edge}, \cite{autonomic}, and~\cite{fog_ga} propose algorithms that generate deployment schemes that provide better QoS or lower cost with resource constraints based on greedy or Markov approximation algorithms. The dependencies between services are always overlooked, and multiple instances for each service are not allowed. For example, the work~\cite{space_air_ground} modeled the mobility-aware joint service placement problem as a Mixed Integer Linear Programming (MILP) optimization, and the work~\cite{winning_at_starting} aimed to minimize the total delay in Mobile Edge Computing (MEC) by deploying each service to one server only. In contrast to this, the work~\cite{User_allocation} proposed a stochastic user allocation algorithm to assign users with different categories of requirements to different servers, which also means desired services should be deployed on servers. In work~\cite{deep_reinforcement_algo} and \cite{hani_2021_deep_reinforcement_fog}, deep reinforcement learning was used to seek an optimized deployment scheme on resource-constrained edges for better QoS, like lower service response time. Multiple instances for each service are allowed, but dependencies are overlooked. Lastly, the work~\cite{predictive_placement} solved the problem considering the prediction of user mobility based on a frame-based design for lower long-term time-average service delay based on Lyapunov optimization. In summary, even though these studies do well on some scenarios, the dependencies between services and multiple instances support should also be considered when solving the deployment problem in microservice systems.

In terms of the multi-component application placement problem, many studies focused on power saving, cost-saving, or better QoS, where the application consists of multiple components that can be deployed on different servers~\cite{GMCAPP,LDSPP,badep,Bahreini_2017_placement,optimal_application_deployment}. There, each component is only deployed on one server at the same time. The deployment schemes are mostly modeled as a matrix $ X $ where $ x_{i,j} \in \{0, 1\} $ stands for component $ i $ should be deployed on server $ j $ or not, and the problems are formulated as an Integer Linear Programming (ILP) problem or Mixed-ILP. The work~\cite{elitism_based} employed a Genetic Algorithm (GA) to solve the problem for minimizing multiple optimization goals, including service time, energy consumption, and service cost with resource constraints. Meanwhile, the work~\cite{wang_2017_edge_online_placement} treated both the application and the physical computing system as two graphs. Online approximation algorithms were proposed based on graph-to-graph placement, and the mapping from the application graph to the physical graph is treated as the deployment scheme. The work~\cite{stochastic_optimization} developed a stochastic optimization approach to maximize QoS in MEC with Markov Decision Processes. In summary, even though the dependencies between components are considered, the constraint that each component can only be deployed on one server at the same time makes these methods less effective in microservice systems, and the missing support of multiple applications overlooks the complex dependencies between services in the microservice systems.

The work~\cite{deng_2020_optimal} considered both the dependencies and multiple instances support in the service placement problem in resource-constrained distributed edges, while aiming to minimize running cost with QoS constraints. However, the proposed algorithm is based on the relaxed ILP problem within the branch and bound method, making it unsuitable for the online system.

In our opinion, it is essential to take the support for service dependencies and multiple instances into consideration in the service placement problem in microservice systems, while being computationally efficient. In a microservice system, most user requirements should be satisfied by a service set instead of one service, making the dependencies necessary at run-time. Multiple instances support for one service also plays an essential role in a microservice system for load balancing. Thus, those two factors must be considered for an online deployment algorithm in the placement problem.

\section{Problem Formulation}\label{sec:problem_formulation}

In this section, we formulate the Service Placement Problem in Microservice Systems (SPPMS), and we use the average response time as the optimization goal. A microservice system consists of a set of \textit{services}, and there are many instances deployed on \textit{servers} for every service. Each of the services has many functions which are exposed as APIs, and dependencies exist due to the requests between services.

\subsection{Microservice System Model}

This section details the definition of \textit{Service}, \textit{Dependency}, \textit{Function Chain}, \textit{Server}, and \textit{Deployment Scheme} in this work.

\begin{sloppypar}
\textbf{Definition 1 (Service)} The service set in the system is described as $ S $, and a service $ s_i $ is defined as a tuple ${ s_i = <\mathcal{F}_i, \mu_i, r^s_i>, s_i \in S }$, where:
\begin{itemize}
    \item $ \mathcal{F}_i = \{f_{i,1}, f_{i,2}, ..., f_{i,n}\} $ denotes a set of functions offered by $ s $. Each of them is described as a pair ${ f_{i,j} = <d^{in}_{i,j}, d^{out}_{i,j}> }$ where $ d^{in}_{i,j} $ and $ d^{out}_{i,j} $ stand for the input and output data size of $ f_{i, j} $.
    \item $ \mu_i $ stands for the processing capacity of $ s_i $ denoting the number of requests the service instance can process per unit time.
    \item $ r^s_i $ represents the set of computing resources used for $ s_i $ to achieve $ \mu_i $, such as CPU, RAM, and hard disk storage.
\end{itemize}
\end{sloppypar}

It should be noted that the APIs of each service $ s_i $ are treated as the function set $ \mathcal{F}_i $. Thus, every API is mapped to one unique $ f_{i, j} $, and the requests between APIs are treated as dependencies between services.

\textbf{Definition 2 (Dependency)} The dependencies between services are represented as the call graph $ DG = <F, E> $ of functions provided by each service. $ DG $ is a directed graph, where:
\begin{itemize}
    \item $ F $ denotes the set of functions provided by all the services.
    \item $ E $ stands for the edges between functions. The weight of each edge represents the Average Call Frequency Coefficient (ACFC) between two functions. An edge from $ f_{i,j} $ to $ f_{m,n} $ means $ f_{m,n} $ is called by $ f_{i,j} $, and ACFC($f_{i,j}$, $f_{m,n}$) shows how many times $ f_{m,n} $ should be called by $ f_{i,j} $ when $ f_{i,j} $ is called once.
\end{itemize}

Note that we do not consider cyclic dependencies in this work, as they are considered a bad practice in other architectures~\cite{Auto_Detection_Smells}, and they can be hard to maintain or reuse in isolation~\cite{Definition_MBS}.
Thus, $ DG $ is a directed acyclic graph. Furthermore, due to the complex logic, ACFC should be calculated according to the request history of the whole service system. With service requests in the if-else clause at the source code level, it is difficult to estimate the ACFC because the probabilities of executing the if-clause and else clauses cannot be predicted accurately.

\textbf{Definition 3 (Function Chain)} The function chain of $ f_{i,j} $ is defined as $ L_{i,j} = <f_{i,j}, ..., f_{m, n}> $. It is used to describe the requests between functions after the system receives a request to $ f_{i,j} $. $ L_{i,j}^m \in L_{i,j} $ stands for the \textit{m}-th function in $ L_{i,j} $.

When users or IoT devices send a request to any API in the system, a set of $ f_{i,j} $ are called, and the calling path between them is a subgraph of $ DG $. In SPPMS, the non-linear calling path can be converted to a function chain as shown in \figurename~\ref{fig:chain_conversion}. By adding a new virtual call between $ f_{3,2} $ and $ f_{2,3} $, which is in red, without input and output data, the calling path graph in the upper part of the figure is converted to the equivalent function chain in the lower part. Besides, the ACFC should also be set to zero which means $ f_{3,2} $ does not call $ f_{2,3} $ in real.

\begin{figure}[!t]
    \centering
    \includegraphics[width=\linewidth]{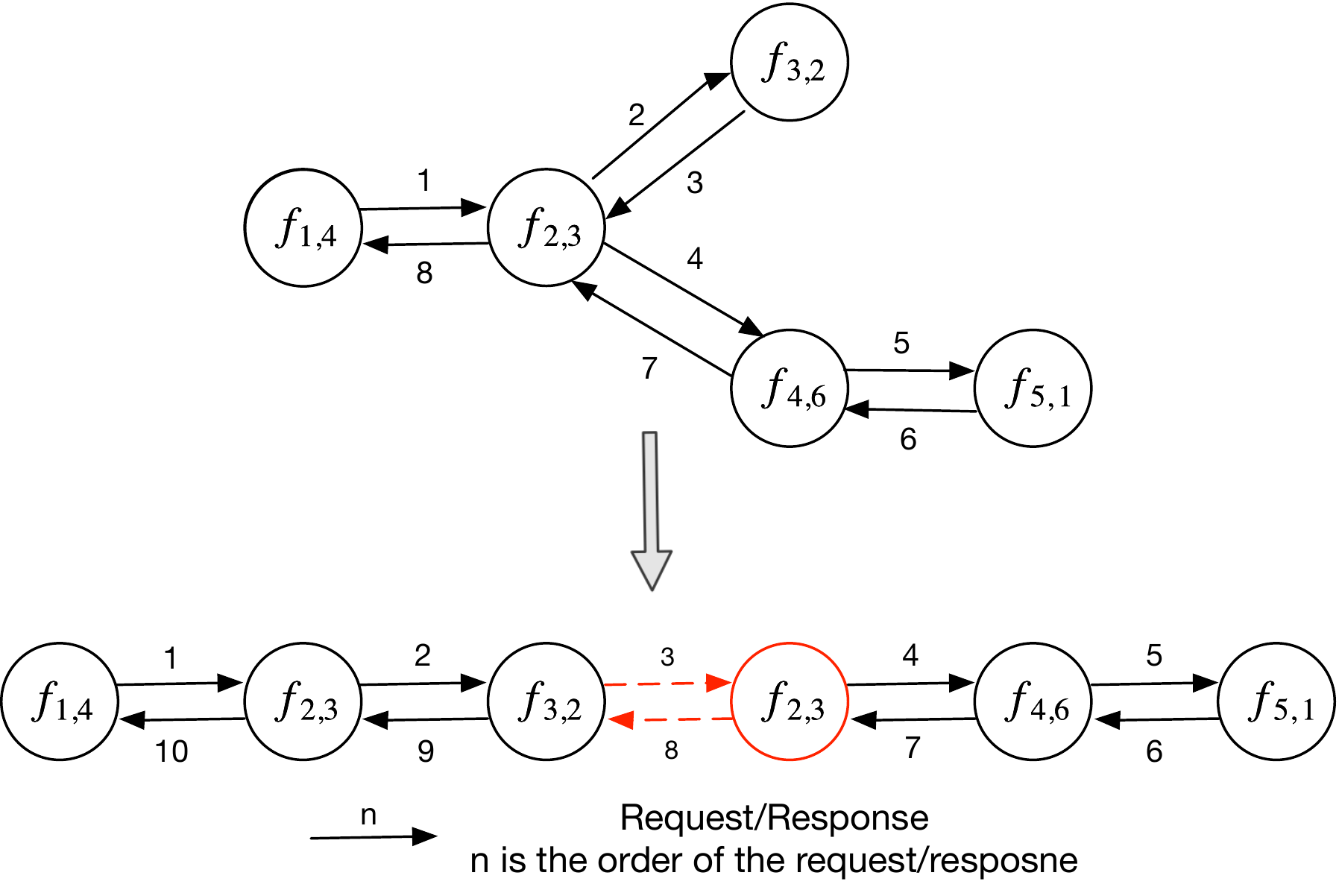}
    \caption{Conversion from function calling graph to chain}
    \label{fig:chain_conversion}
\end{figure}

As the non-linear calling path can be converted to a function chain as detailed above, it does not impose additional constraints on service dependencies.


\textbf{Definition 4 (Server)} The set of servers in the system is denoted as $ N $. For $ \forall n_i \in N $, $ n_i = <r_i^n> $, where $ r_i^n $ describes the computing resources of server $ n_i $. The delay and bandwidth between $ n_i $ and $ n_j $ are defined as $ d_{i,j} $ and $ b_{i,j} $, respectively.

Because not all nodes are fully connected to others, the total delay and the lowest bandwidth are used according to the routing path of any two nodes in the real system, which are determined by the routing rules.

\textbf{Definition 5 (Deployment Scheme)} The deployment scheme is described as a matrix $ X = \{x_{i,j} | x_{i,j} \in \mathbb{N}, 1 \leq i \leq |N|, 1 \leq j \leq |S|\} $, where $ x_{i, j} $ stands for the instance count of service $ s_i $ on server $ n_j $.

Table~\ref{tab:notation} summarizes all the symbols and notations in this paper for convenience.

\begin{table}
\caption{Notation}
\label{tab:notation}
\centering
\begin{tabular}{ll} 
\hline
Notation & Description  \\ 
\hline
    $ s_i $                             &   Microservice $ s_i, s_i \in S $                 \\
    $ \mathcal{F}_i $                   &   Function set of $ s_i $                         \\
    $ f_{i, j} $                        &   One of the function of service $ s_i $, $ f_{i,j} \in \mathcal{F}_i $          \\
    $ d^{in}_{i,j},\ d^{out}_{i,j} $    &   Input and output data size of $ f_{i, j} $      \\
    $ \mu_i $                           &   Processing capacity of $ s_i $                  \\
    $ r^s_i $                           &   Computing resource set that used by $ s_i $     \\
    $ F $                               &   All the function set provided by all services \\
    $ DG $                              &   Dependency graph of the service system          \\
    ACFC($f_{i,j}$, $f_{m,n}$)          &   \begin{tabular}[c]{@{}l@{}}Average call frequency coefficient between\\$ f_{i,j} $ and $ f_{m,n} $\end{tabular}  \\
    $ L_{i,j} $                         &   Function chain of $ f_{i,j} $                   \\
    $ L_{i,j}^m $                       &   The m-th function of $ L_{i,j} $                \\
    $ n_i $                             &   Server node $ n_i, n_i \in N $                  \\
    $ d_{i,j},\ b_{i,j} $               &   Delay and bandwidth between $ n_i $ and $ n_j $        \\
    $ X $                               &   Deployment scheme                               \\
    $ x_{i,j} $                         &   Instance count of $ s_i $ on $ n_j $            \\
    $ h_{i,j} $                         &   \begin{tabular}[c]{@{}l@{}}A response server path for $ L_{i,j} $, $ h_{i,j} \in H_{i,j} $ \end{tabular}    \\
    $ \lambda_{i,j}^k $                 &   Request rate of $ f_{i,j} $ on $ n_k $          \\
    $ T $                               &   The average response time of the system         \\
    $ C_{max} $                         &   Maximum deployment cost                         \\
    $ \gamma_i^u $                      &   Total request rate of $ s_i $ from users        \\
    $ \gamma_i^s $                      &   Total request rate of $ s_i $ from services     \\
\hline
\end{tabular}
\end{table}

\subsection{Problem Definition}

\begin{figure*}[!ht]
    \centering
    \includegraphics[width=0.9\linewidth]{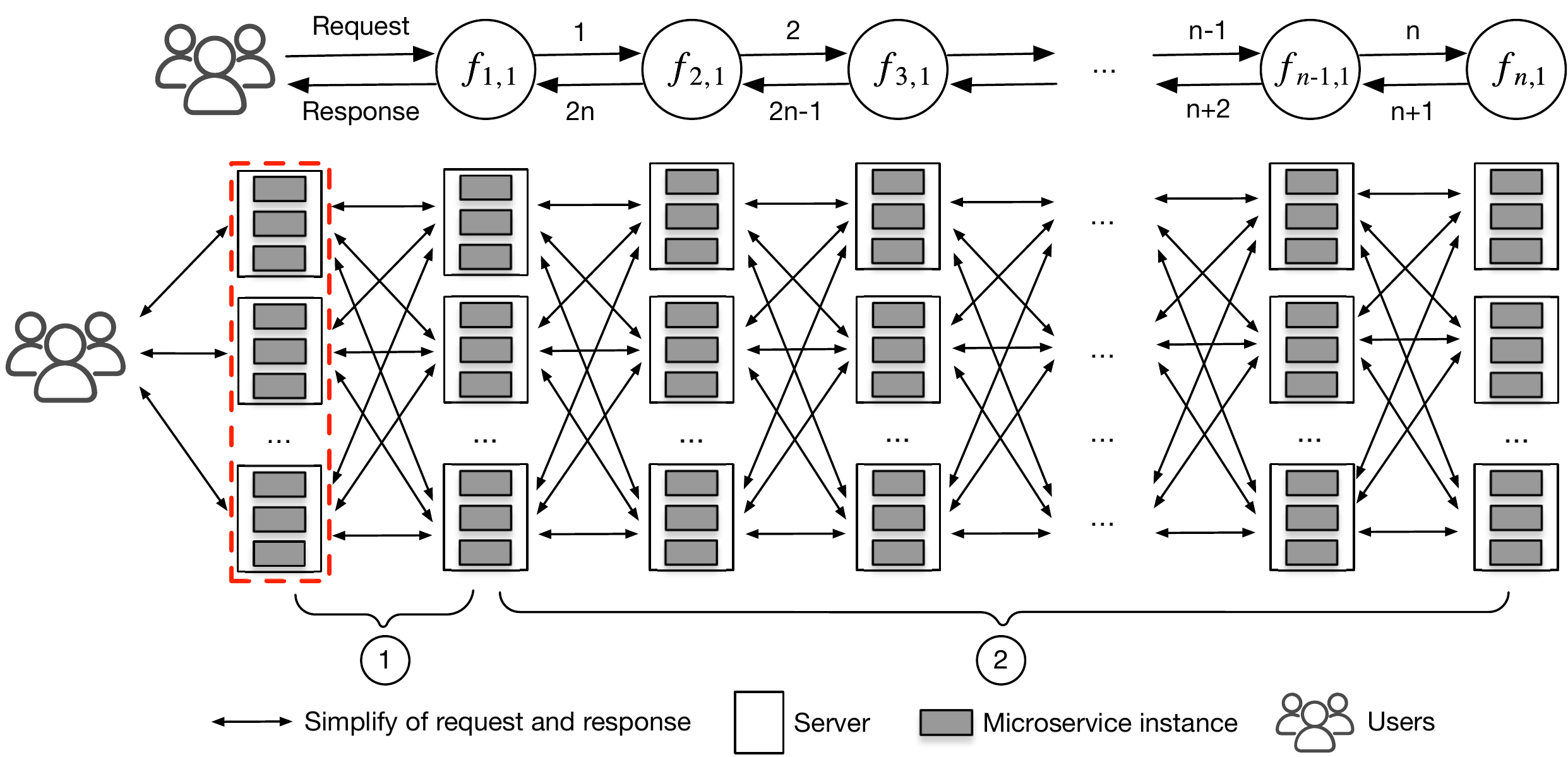}
    \caption{An overview of the request routing for any $ f_{i,j} $}
    \label{fig:routing_overview}
\end{figure*}

This section provides the formulation of the Service Placement Problem in Microservice Systems (SPPMS) with the optimization goal and constraints.

\subsubsection{Average Response Time}

When the microservice system responds to a request to $ f_{i,j} $, it must select one service instance for each function in $ L_{i,j} $ due to the dependencies between services. $ H_{i,j} $ denotes all the possible response server paths, and $ h_{i,j} \in H_{i,j} $ stands for one selected server path for processing $ L_{i,j} $. $ h^m_{i,j} \in h_{i,j} $ is used to denote the server where the service instance that receives the request to $ L^m_{i,j} $ locates. Also, $ |h_{i,j}| = |L_{i,j}| $.

There are many load balancing algorithms used in microservices system with multiple service instances support, such as Round Robin, Weighted, and Random. Note that we consider the commonly used Round Robin routing in our SPPMS formulation; an adaptation to many other load balancing algorithms is straightforward.

The primary process of the request routing in a microservice system is shown in \figurename~\ref{fig:routing_overview}. The transmission between users and servers is not considered as it depends on many factors outside of the system, and we also assume the requests from the users are sent to the nearest servers first, as the servers in the red rectangle shown in \figurename~\ref{fig:routing_overview}. It should be noted that the \figurename~\ref{fig:routing_overview} is based on synchronous requests. However, it is also suitable for asynchronous requests by setting the output data size of each function to 0.

Under the Round Robin routing algorithm, all the service instances of service $ s_i $ in the system have the same probability of receiving the request about the function $ f_{i,j} $ when there is a sufficiently large number of requests. In this paper, we assume that the number of requests is large enough and the probability is the same. Thus, the probability of processing the request about $ f_{i,j} $ on $ n_{k} $ is:

\begin{equation}
    Prob(f_{i,j}, n_k) = \frac{x_{k,i}}{\sum_{m=1}^{|N|} x_{m,i}}
\end{equation}

For any given selected response service path $ h_{i,j} $, the probability of selecting it from $ H_{i,j} $ is:

\begin{equation}\label{eqa:prob_path}
    Prob(h_{i,j}) = \prod_{k=1}^{|L_{i,j}|} Prob(L^k_{i,j}, h^k_{i,j})
\end{equation}

The probability of the request of $ f_{i,j} $ from server $ n_k $ is:

\begin{equation}
    Prob(\lambda_{i,j}^k) = \frac{\lambda_{i,j}^k}{\sum_{m=1}^{|N|} \lambda_{i,j}^m}
\end{equation}

When calculating the average response time, the request processing time after receiving the request is not included to reduce the impact of each microservice implementation. Only the delay and data transmission time are included in this paper. The response time of the request of $ f_{i,j} $ from users near server $ n_k $, which is denoted as $ t_{h_{i,j},k} $, consists of two parts (see \ding{172} and \ding{173} in \figurename~\ref{fig:routing_overview}). The first part is the response time between the request received server and the first of the selected response server path, and the second part is the response time for processing the function chain. $ t_{h_{i,j},k} $ is formulated as:

\begin{equation}\label{eqa:t_p_k}
\begin{aligned}
    t_{h_{i,j},k} = \sum_{l_{v,w}^m, m=2}^{|L|} \left(\frac{d^{in}_{v,w} + d^{out}_{v,w}}{b_{p^{m-1}_{i,j}, p^m_{i,j}}} + d_{p^{m-1}_{i,j}, p^m_{i,j}}\right) \\ + \left(\frac{d^{in}_{i,j} + d^{out}_{i,j}}{b_{k, p^1_{i,j}}} + d_{k, p^1_{i,j}}\right)
\end{aligned}
\end{equation}

Thus, the average response time of $ f_{i,j} $ in the whole system is described as $ T_{i,j} $:

\begin{equation}\label{eqa:T_i_j}
    T_{i,j} = \sum_{k=1}^{|N|} \sum_{h_{i,j}\in H_{i,j}}  Prob(\lambda_{i,j}^k) \times Prob(h_{i,j}) \times t_{h_{i,j},k}
\end{equation}

Consequently, the system's average response time is defined as $ T $ when giving the deployment scheme $ X $:

\begin{equation} \label{eqa:Tx}
    T(X) = \sum_{f_{i,j} \in F} \frac{\sum_{m=1}^{|N|} \lambda_{i,j}^m}{ \sum_{f_{v,w} \in F} \sum_{m=1}^{|N|} \lambda_{v,w}^m } \times T_{i,j}
\end{equation}

\subsubsection{Constraints}

There are three types of constraints in SPPMS: the computing resource constraint of servers, the deployment cost constraint, and the service capability constraint.

For the computing resources, the used resources by service instances should not exceed the resources provided by each server:

\begin{equation} \label{eqa:c1}
    \sum_{i=1}^{|S|} x_{j,i} \times r^s_i \leq r^n_j, \forall n_j \in N
\end{equation}

The deployment cost should also not exceed the given maximum cost $ C_{max} $, and the monetary cost is considered as the deployment cost here. Because most of the cloud providers support to charge based on resource usage, the deployment constraint is defined as follows, where $ c $ is the cost of one unit used resource:

\begin{equation} \label{eqa:c2}
    c \times \sum_{i=1}^{|S|} \sum_{j=1}^{|N|} r^s_{i} \times x_{j,i} \leq C_{max}
\end{equation}

The last constraint is the service capability constraint, which makes sure that there are sufficient service instances to process the requests from the users and other services. The total desired capability of $ s_{i} $ contains desired capability from the users denoted as $ \gamma^u_{i} $ and from the services denoted as $ \gamma^s_{i} $. For the user part, it can be calculated by:

\begin{equation}
    \gamma^u_{i} = \sum_{f_{j,k} \in \mathcal{F}_{i}} \sum_{m=1}^{|N|} \lambda_{j,k}^m
\end{equation}

The service part $ \gamma^s_{i} $ is raised by the requests between services due to dependencies, and it can be calculated directly with $ L_{i,j} $ and $ DG $. For example, assuming $ L_{1,1} = <f_{1,1}, f_{2,1}, f_{3,1}> $, $ ACDC(f_{1,1}, f_{2,1}) = 2 $, $ ACDC(f_{2,1}, f_{3,1}) = 1.5 $ and $ \lambda_{1,1} = 5 $, then $ \gamma^s_{2} = \lambda_{1,1} \times ACDC(f_{1,1}, f_{2,1}) = 10 $, and $ \gamma^s_{3} = \lambda_{1,1} \times ACDC(f_{1,1}, f_{2,1}) \times ACDC(f_{2,1}, f_{3,1}) = 15 $.

The service capability constraint is formulated as:

\begin{equation} \label{eqa:c3}
    \mu_i \times \sum_{j=1}^{|N|} x_{j, i} \geq \gamma^u_{i} + \gamma^s_{i} , \forall s_i \in S
\end{equation}

\subsubsection{Optimization Problem Definition}

By combining Equations~\ref{eqa:Tx}, \ref{eqa:c1}, \ref{eqa:c2}, and \ref{eqa:c3}, we can define the optimization problem.

Given the service set $ S $, the dependency graph $ DG $, the server set $ N $, user requirements about each function on every server $ \lambda_{i,j}^k, \forall f_{i,j} \in F, \forall n_k \in N $, and maximum cost $ C_{max} $, the goal is to find a deployment scheme $ X $ that satisfies:

\begin{subequations}
\begin{gather}
    min \ T(X), \ X \in \mathbb{N}_{|N| \times |S|} \\
s.t.\left\{
\begin{array}{lr}
    \sum_{i=1}^{|S|} x_{j,i} \times r^s_i \leq r^n_j, & \forall n_j \in N   \\
    c \times \sum_{i=1}^{|S|} \sum_{j=1}^{|N|} r^s_{i} \times x_{j,i} \leq C_{max}   \\
    \mu_i \times \sum_{j=1}^{|N|} x_{j, i} \geq \gamma^u_{i} + \gamma^s_{i}, & \forall s_i \in S
\end{array}
\right.
\end{gather}
\end{subequations}

Because both the denominator and the numerator of $ T(X) $ are multivariate polynomials~\cite{solving_fpp}, the problem is a Fractional Polynomial Problem (FPP).

\section{Algorithms}\label{sec:algorithm}

In the following, we show how to convert our FPP to a Quadratic Sum-of-Ratios Fractional Problem (QSRFP), which can significantly reduce the computational complexity comparing to FPP. Let us assume that the average length of function chain is $ \hat{l} $. Then, to evaluate the average response time of any function, FPP needs to iterate over all the possible response path, whose size is $ |N|^{\hat{l}} $, while QSRFP does not need to. QSRFP is the problem whose constraint functions and the denominator and numerator of the optimization goal are all quadratic functions, which are possibly not convex~\cite{solving_qsrfp}. Two efficient greedy-based algorithms are detailed to solve QSRFP, and the optimal algorithm for solving QSRFP is also presented for reference. In general, a Quadratic Sum-of-Ratios Fractional Programs (QSRFP) problem has the following structure~\cite{solving_qsrfp}:
\begin{subequations} \label{eqa:fpp}
\begin{gather}
{\rm min} \ H_0(y) = \sum_{i=1}^p \frac{f_i(y)}{g_i(y)}  \\
s.t.\left\{
\begin{array}{lr}
    H_m(y) \leq 0, m = 1, ..., M,   \\
    y \in Y^0 = \{y \in R^n : \underline{y}^0 \leq y \leq \overline{y}^0\} \subset R^n
\end{array}
\right.
\end{gather}
\end{subequations}

where $ p \leq 2 $; $ f_i(y) $, $ g_i(y) $, and $ H_m(y) $ are all quadratic functions, i.e. $ f_i(y) = \sum_{j=1}^n \sum_{k=1}^n \delta^i_{jk} y_i y_k + \sum_{k=1}^n c_k^i y_k + \overline{\delta}_i $.

\subsection{Converting FPP to QSRFP}\label{subsec:converting}

Although there are already some optimization approaches for the FPP~\cite{solving_fpp,optimization_fpp,global_minimization_fpp}, it is difficult to apply them as online algorithms due to the high computational complexity of $ T(X) $. 
Assuming the average length of function chain is $ \hat{l} $, there are $ |N|^{\hat{l}} $ possible response server paths for each function, and Equation~\ref{eqa:prob_path} must be calculated for $ |F| \times |N|^{\hat{l}} $ times, which is unaffordable for an online algorithm. Thus, we cannot consider FPP as defined above. To simplify the problem, we prove the following theorem:

\begin{theorem}\label{theorem}
Given ${ <f_1, f_2, ..., f_{k-1}, f_{k}, ..., f_{n}> }$ as the target function chain, let $ T_{i \rightarrow j} $ denote the average response time from $ f_i $ to $ f_j $. Then we have: 

\begin{equation}
    T_{1 \rightarrow n} = \sum_{i=1}^{n-1}T_{i \rightarrow i+1}
\end{equation}
\end{theorem}

\begin{proof}

 Define $ H_{i \rightarrow j} $ as all the possible server paths from $ f_{i} $ to $ f_{j} $ and $ h_{i \rightarrow j} $ as a server path $ \forall h_{i \rightarrow j} \in H_{i \rightarrow j} $, $ P_{h_{i \rightarrow j}} $ denotes the probability of $ h_{i \rightarrow j} $, and $ t_{h_{i \rightarrow j}} $ represents the response time of $ h_{i \rightarrow j} $. Then $ \forall k \in [1, n] $, we have :
 
 \begin{equation}\label{eqa:proof_T}
     T_{1 \rightarrow n} = \sum_{h_{1 \rightarrow n} \in H_{1 \rightarrow n}} P_{h_{1 \rightarrow n}}t_{h_{1 \rightarrow n}}
 \end{equation}
 
  \begin{equation}\label{eqa:proof_P}
     \sum_{h_{1 \rightarrow n} \in H_{1 \rightarrow n}} P_{h_{1 \rightarrow n}} = \sum_{h_{1 \rightarrow k} \in H_{1 \rightarrow k}} \sum_{h_{k \rightarrow n} \in H_{k \rightarrow n}} P_{h_{1 \rightarrow k}} P_{h_{k \rightarrow n}}
 \end{equation}
 
 \begin{equation}\label{eqa:proof_t}
     t_{h_{1 \rightarrow n}} = t_{h_{1 \rightarrow k}} + t_{h_{k \rightarrow n}}
 \end{equation}
 
 After applying Equation \ref{eqa:proof_P} and \ref{eqa:proof_t} to \ref{eqa:proof_T}, $ \forall k \in [1, n] $:

\begin{equation}\label{eqa:proof}
    \begin{aligned}
T_{1 \rightarrow n} & = \sum_{h_{1 \rightarrow n} \in H_{1 \rightarrow n}} P_{h_{1 \rightarrow n}}t_{h_{1 \rightarrow n}} \\
& = \sum_{h_{1 \rightarrow n} \in H_{1 \rightarrow n}} P_{h_{1 \rightarrow n}}t_{h_{1 \rightarrow k}} + \sum_{h_{1 \rightarrow n} \in H_{1 \rightarrow n}} P_{h_{1 \rightarrow n}}t_{h_{k \rightarrow n}} \\
& = \sum_{h_{k \rightarrow n} \in H_{k \rightarrow n}} \left(\sum_{h_{1 \rightarrow k} \in H_{1 \rightarrow k}} P_{h_{1 \rightarrow k}} t_{h_{1 \rightarrow k}}\right) P_{h_{k \rightarrow n}} \\
& \qquad + \sum_{h_{1 \rightarrow k} \in H_{1 \rightarrow k}} \left(\sum_{h_{k \rightarrow n} \in H_{k \rightarrow n}} P_{h_{k \rightarrow n}} t_{h_{k \rightarrow n}}\right) P_{h_{1 \rightarrow k}} \\
& = \sum_{h_{k \rightarrow n} \in H_{k \rightarrow n}} P_{h_{k \rightarrow n}} T_{1 \rightarrow k} + \sum_{h_{1 \rightarrow k} \in H_{1 \rightarrow k}} P_{h_{1 \rightarrow k}} T_{k \rightarrow n} \\
& = T_{1 \rightarrow k} + T_{k \rightarrow n}
\end{aligned}
\end{equation}

After applying Equation \ref{eqa:proof} repeatedly, we have:

\begin{equation}
\begin{aligned}
    T_{1 \rightarrow n} & = T_{1 \rightarrow 2} + T_{2 \rightarrow n}  \\
    & = T_{1 \rightarrow 2} + T_{2 \rightarrow 3} + T_{3 \rightarrow n} \\
    & = ... \\
    & =T_{1 \rightarrow 2} + T_{2 \rightarrow 3} + ... + T_{n-1 \rightarrow n} \\
    &= \sum_{i=1}^{n-1}T_{i \rightarrow i+1}
\end{aligned}
\end{equation}

\end{proof}

To make the formula more concise, the users can be treated as $ L_{i,j}^0 $, and $ L_{i,j} $ is abbreviated as $ L $. Then, using Theorem~\ref{theorem},  $ T_{i,j} $ in Equation~\ref{eqa:T_i_j} can be represented as:

\begin{equation}\label{eqa:q_t_i_j}
\begin{aligned}
    T_{i,j} = \sum_{k=0}^{|L| - 1} \underbrace{\sum_{v}^{|N|} \sum_{w}^{|N|} (\overbrace{\frac{x_{v,L^k}}{\sum_{m}^{|N|} x_{m,L^k}} \times \frac{x_{w,L^{k+1}}}{\sum_{m}^{|N|} x_{w,L^{k+1}}}}^{prob \ of \ L^k \ on \ n_j \ and \ L^{k+1} \ on \ n_w} 
    \times t_{v,w,k+1})}_{avgTime\ between\ L^k\ and\ L^{k+1}\ as\ a\ quadratic\ fraction}
\end{aligned}
\end{equation}

where $ t_{v,w,k+1} = \frac{d^{in}_{L^{k+1}} + d^{out}_{L^{k+1}}}{b_{v,w}} + d_{v,w} $ denotes the response time between $ L^{k} $ on $ n_w $ and $ L^{k+1} $ on $ n_w $. By defining $ \hat{\lambda}_{i,j} =  \frac{\sum_{m=1}^{|N|} \lambda_{i,j}^m}{ \sum_{f_{v,w} \in F} \sum_{m=1}^{|N|} \lambda_{v,w}^m } $, $ T(X) $ can be described as

\begin{equation}\label{eqa:TX_qsrfp}
    T(X) = \sum_{f_{i,j} \in F} \hat{\lambda}_{i,j} \times T_{i,j}
\end{equation}

It should be noted that $ T_{i,j} $ is a quadratic sum-of-rations fraction, and the same as $ T(X) $ because $ \hat{\lambda}_{i,j} $ is a constant. Thus, we have converted the FPP to a QSRFP. Compared with the FPP in Equation~\ref{eqa:fpp}, there is no need to iterate over all possible response paths for each function chain, which reduces the computational cost significantly.

\subsection{Greedy-based Algorithms}

There are two greedy-based algorithms\footnote{https://github.com/HIT-ICES/AlgoDeployment} we proposed: B-QSRFP in Algorithm~\ref{algo:bfs_placement} and D-QSRFP in Algorithm~\ref{algo:dfs_placement}. They depend on two sub-algorithms~\ref{algo:best_server} and~\ref{algo:deploy} to find the best server for each service and deploy service instances.
The greedy-based algorithms are inspired by Equation~\ref{eqa:q_t_i_j}: based on it, the best placement of each service $ s_i $ is only related to the services that call $ s_i $ and the services called by $ s_i $, and the response server paths $ H $ has no influence on $ T_{i,j} $. Algorithm~\ref{algo:best_server} uses this and selects the best server for new instances of $ s_j $ based on Equation~\ref{eqa:q_t_i_j}.

This algorithm only considers the servers that have enough computing resources for at least one instance of $ s_j $ as shown in lines 3--4. For each server $ n $, the algorithm evaluates the average response time based on Equation~\ref{eqa:T_i_j} when a new instance of $ s_j $ is deployed on $ n $ as shown in lines 5--11, where $ DG.{\rm pred}(s_j) $ and $ DG.{\rm succ}(s_j) $ denotes the predecessors and successors of $ s_j $ in $ DG $. The response time between $ s_j $ and $ s_p $ equals 0 if there is no instance of $ s_p $ during the evaluation. Lines 12--14 selects the server with the smallest average response time as the best server for $ s_j $.

When deploying an instance of $ s_i $, the best servers for the predecessors and successors are changed according to Equation~\ref{eqa:q_t_i_j}, which means the algorithm needs to replace the instances of the predecessors and successors, as detailed in Algorithm~\ref{algo:deploy}.

\begin{algorithm}[H]
\caption{bestServer($s_j $, $ X $, $ N $, $ DG $) Algorithm}\label{algo:best_server}
\textbf{Input:} $ s_j $: service wanted to deploy \\
\hspace*{6ex}$ X $: current deployment scheme \\
\hspace*{6ex}$ N $: server list \\
\hspace*{6ex}$ DG $: dependency graph \\
\textbf{Output:} best node to deploy $ s_j $
\begin{algorithmic}[1]
\State $ s \gets None $
\State $ c \gets +\infty $
\For{$ n $ in $ N $}
    \If{ $ n $ has enough resource for $ s_j $ }
        \State $ currC \gets 0 $
        \For{ $ s_p $ in $ DG.{\rm pred}(s_j) $ }
            \State $ currC \gets currC + {\rm avgTime}(s_j, s_p, n, X) $ \Comment{avgTime: The avgTime between $ s_j $ and $ s_p $ based on Equation~\ref{eqa:T_i_j} after deploying an instance of $s_j$ on $n$.}
        \EndFor
        \For{ $ s_s $ in $ DG.{\rm succ}(s_j) $ }
            \State $ currC \gets currC + {\rm avgTime}(s_j, s_s, n, X) $
        \EndFor
        
        \If{ $ currC < c $ }
            \State $ s,\ c \gets n,\ currC $
        \EndIf
    \EndIf
\EndFor
\State \Return $ s $
\end{algorithmic}
\end{algorithm}

\begin{algorithm}[H]
\caption{DeploySpread($s_i$, $k$, $ X $, $ N $, $ DG $) Algorithm}\label{algo:deploy}
\textbf{Input:} $ s_j $: service wanted to deploy \\
\hspace*{6ex}$ k $: how many instances wanted to deploy \\
\hspace*{6ex}$ X $: current deployment scheme \\
\hspace*{6ex}$ N $: server list \\
\hspace*{6ex}$ DG $: dependency graph
\begin{algorithmic}[1]
\While{ $ k > 0 $ }
    \State $ n \gets {\rm bestServer}(s_j, X, N, DG) $
    \State $ c \gets {\rm min}(\lfloor \frac{{\rm leftRes}(n)}{r_j^s} \rfloor, k ) $
    \State $ {\rm deploy}(s_i, n, c) $
    \State $ k \gets k - c $
    
    \State $ S_r \gets DG.{\rm pred}(s_j) \cup DG.{\rm succ}(s_j) $
    \State $ S_p \gets \{ s_j \} $
    \State $ X_{copy} \gets X $
    \While{ len$ (S_r) \ne 0 $ }
        \State $ s_r \gets S_r.{\rm pop}() $
        \State $ k_r \gets {\rm instNum}(s_r, X) $
        \State delete all the instances of $ s_r $ from X
        \While{ $ {\rm instNum}(s_r, X) < k_r $ }
            \State $ n_r \gets {\rm bestServer}(s_r, X, N, DG) $
            \State $ c_r \gets {\rm min}(\lfloor \frac{{\rm leftRes}(n_r)}{r_r^s} \rfloor, k_r - {\rm instNum}(s_r, X) ) $
        \EndWhile
        \If{ $ X \ne X_{copy} $ }
            \State $ S_r \gets S_r \cup DG.{\rm pred}(s_r) \cup DG.{\rm succ}(s_r) \setminus S_p $
            \State $ S_p \gets S_p \cup \{ s_r \} $
        \EndIf
    \EndWhile
\EndWhile
\end{algorithmic}
\end{algorithm}

\begin{algorithm}[h]
\caption{BFS Placement Algorithm: B-QSRFP}\label{algo:bfs_placement}
\textbf{Input:} $ \Gamma^s $:  $ \{ \gamma_{i}^s | \forall s_{i} \in S \} $  \\
\hspace*{6ex}$ \Gamma^u $:  $ \{ \gamma_{i}^u | \forall s_{i} \in S \} $  \\
\hspace*{6ex}$ N $: server list \\
\hspace*{6ex}$ DG $: dependency graph \\
\textbf{Output:} Deployment Scheme
\begin{algorithmic}[1]
\State $ X \gets \mathbf{0} $
\While{ not $ DG.{\rm empty}() $ }
    \State $ S_c \gets \{ s_i | s_i \in S, {\rm len}(DG.{\rm pred}(s_{i})) = 0 \} $
    \State $ s_c \gets {\rm argmin}_{s_i \in S_c} \frac{\mu_i}{r_i^s} $
    \State $ {\rm DeploySpread}(s_i, \lceil \frac{\gamma_i^s + \gamma_i^u}{\mu_i} \rceil, X, N, DG) $
    \State $ DG.{\rm remove}(s_i) $
\EndWhile
\State \Return $ X $
\end{algorithmic}
\end{algorithm}

Algorithm~\ref{algo:deploy} is responsible for deploying $ k $ instances of $ s_i $ in the system considering the dependencies between services. After finding the best server for $ s_i $ with Algorithm~\ref{algo:best_server}, instances are deployed on the server as many as needed, as shown in lines 2--4. Lines 6--21 re-deploy instances of predecessors and successors. The instances of them are deleted and re-deployed based on Algorithm~\ref{algo:best_server} again in lines 12--16. The re-deployment is continued for the services whose deployment is changed after re-deployment in lines 17--20. The re-deployment is stopped when all the affected services are processed.

\begin{algorithm}[H]
\caption{DFS Placement Algorithm: D-QSRFP}\label{algo:dfs_placement}
\textbf{Input:} $ L_u $: $ \{ L(f_{i,j}) | \sum_m^{|N|} \lambda_{i,j}^m > 0 \} $ \\
\hspace*{6ex}$ \Lambda $:  $ \{ \lambda_{i,j}^k | \forall f_{i,j} \in F, \forall k \in N \} $  \\
\hspace*{6ex}$ N $: server list \\
\hspace*{6ex}$ DG $: dependency graph \\
\textbf{Output:} Deployment Scheme
\begin{algorithmic}[1]
\State $ X \gets \mathbf{0} $
\State $ \lambda_{solved} \gets {\rm a\ dict\ whose\ values\ are\ 0}\ \forall s_i \in S $
\While{ len$(L_u) \ne 0 $ }
    \State $ l_u \gets {\rm maximum\ data\ transmission\ chain\ from\ L_u} $
    \State $ coe \gets 1 $
    \State $ f_{prev} \gets None $
    \For{ $ f_{i,j} $ in $ l_u $ }
        \If{ $ f_{prev} \ne None $ }
            \State $ coe \gets coe \times {\rm ACDC}(f_{prev}, f_{i,j}) $
        \EndIf
        \State $ f_{prev} = f_{i,j} $
        \State $ \lambda_c = coe \times \sum_m^{|N|}\lambda_{i,j}^m $
        \State $ \lambda_d = \lambda_{solved}[s_i] + \lambda_c - {\rm instNum}(s_i) \times \mu_i $
        \If{ $ \lambda_d > 0 $ }
            \State $ {\rm DeploySpread}(s_i, \lceil \frac{\lambda_d}{\mu_i} \rceil, X, N, DG) $
            \State $ \lambda_{solved}[s_i] \gets \lambda_{solved}[s_i] + \lambda_c $
        \EndIf
    \EndFor
    \State $ L_u \gets L_u \setminus \{ l_u \} $
\EndWhile
\State \Return $ X $
\end{algorithmic}
\end{algorithm}

Based on Algorithm~\ref{algo:deploy}, we propose two greedy algorithms based on depth-first search (DFS) and breadth-first search (BFS); see Algorithms~\ref{algo:bfs_placement} and~\ref{algo:dfs_placement}. The difference between them is the order of the service deployment. For BFS Placement Algorithm~\ref{algo:bfs_placement}, the service that has the best service ability with minimum computing resources and without predecessors is selected as shown in lines 3--4. Line 5 deploys the minimum number of instances which can satisfy the requests from users and services. After deploying $ s_i $, it is removed from $ DG $ in line 6, and the algorithm stops when all services are deployed as line 2.

The DFS Placement Algorithm~\ref{algo:dfs_placement}, on the other hand, deploys the function chain with maximum data transmission size first as line 4, and the services are deployed following the order of the function chain $ l_u $. Then, the algorithm calculates the request rate $ \lambda_c $ of each service in $ l_u $ and deploys the minimum number of instances to satisfy the desired ability $ \lambda_d $ in lines 8--17. Finally, it ends when all $ f_{i,j} $ that the users request are processed.

\subsection{The Optimal Algorithm}

This section briefly introduces the optimal algorithm for solving QSRFP, which mainly comes from~\cite{solving_qsrfp}.
It serves as a reference only, because its high computational cost prevents us from using it as an online algorithm.

The basic idea of the solution in~\cite{solving_qsrfp} is constructing a parametric relaxation linear programming problem (PRLP) of QSRFP in $ Y $:

\begin{subequations}
\begin{gather}
{\rm min} \ H_0^L(y) = \sum_{i=1}^p \frac{f_i^L(y)}{\overline{g}_i^U}  \\
s.t.\left\{
\begin{array}{lr}
    H^L_m(y) \leq 0, m = 1, ..., M,   \\
    y \in Y^0 = \{y \in R^n : \underline{y}^0 \leq y \leq \overline{y}^0\} \subset R^n
\end{array}
\right.
\end{gather}
\end{subequations}

where $ H_0^L \leq H_0 $, $ H_m^L \leq H_m $, $ f_i^L \leq f_i $, and $ \overline{g_i}^U = {\rm max}_{y \in Y} g_i^U $. For more details about how constructing PRLP, please refer to Section~2 in~\cite{solving_qsrfp}. With Theorems~1 and~2 in~\cite{solving_qsrfp}, $ H_0^L \rightarrow H_0(y) $ as $ ||\overline{y} - \underline{y}|| \rightarrow 0 $.

The branching technique is also used for searching the optimal solution by iteratively subdividing the rectangle $ Y^k $ into two sub-rectangles. For any selected sub-rectangle $ Y^k = [\underline{y}, \overline{y}] \subseteq Y^0$, The sub-rectangles $ Y^{k,1} $ and $ Y^{k,2} $ are generated by dividing $ y_\theta $ into $ \left[\underline{y}_\theta, \frac{\underline{y}_\theta + \overline{y}_\theta}{2}\right]$ and $ [\frac{\underline{y}_\theta + \overline{y}_\theta}{2}, \overline{y}_\theta ]$, where $ \theta \in {\rm argmax}\{ \overline{y}_j - \underline{y}_j: j = 1, ..., n \} $. The Theorem~3 in~\cite{solving_qsrfp} shows it is possible to apply reducing technique by solving the PRLP in $ Y^{k,1} $ and $ Y^{k,2} $ to determine whether the optimal solution exists in $ Y^{k,1} $ and $ Y^{k,2} $.

The branch-and-bound algorithm is detailed as Algorithm~\ref{algo:bnb}, and $ LB $ and $ UB $ are the lower bound and upper bound of the optimal solution. With bisection in line 10, the algorithm reduces the search space with PRLP (reducing) in line 12 and updates $ LB $ and $ UB $ (bounding) in lines 13-24. The branch-and-bound search is used to search integer solutions as $ y^k $ is a non-integer solution in line 31. For more details, we refer the interested reader to Section~3 in~\cite{solving_qsrfp}.

\begin{algorithm}[t]
\caption{QSRFP Branch-and-Bound Algorithm}\label{algo:bnb}
\textbf{Input:} $ \epsilon $: termination error  \\
\hspace*{6ex}$ Y^0 $: initial solution space
\begin{algorithmic}[1]
\State $ UB_0 = +\infty $
\State $ y^0, LB_0 \gets {\rm solving \ PRLP}(Y^0) $
\If{ $ y^0 $ is feasible to QSRFP }
    \State $ UB_0 \gets {\rm min}\{ H_0(y_0), UB_0 \} $
\EndIf

\If{ $ UB_0 - LB_0 > \epsilon $ }
    \State $ \Pi_0 \gets \{ Y^0 \}, \ F \gets \emptyset, \ y^k \gets y^0 $
    
    \For{ $ k $ in 1,2,... }
        \State $ UB_k \gets UB_{k-1}, \ F \gets F \cup \{ Y^{k-1} \} $
        \State $ Y^{k,1}, Y^{k,2} \gets {\rm subdivides}\  Y^{k-1} $
        \For{ $ t $ in $ \{ 1, 2\} $}
            \State $ Y^{k,t} \gets {\rm reducing} \ Y^{k,t} $
            \State $ y^{k,t}, LB^{k,t} \gets {\rm solving \ PRLP}(Y^{k,t}) $
            \If{ $ y^{mid} $ of $ Y^{k,t} $ is feasible to QSRFP }
                \State $ UB_k \gets {\rm min}\{ H_0(y^{mid}), UB_k \}, y^k \gets y^{mid} $
            \EndIf
            \If{ $ y^{k,t} $ is feasible to QSRFP }
                \State $ UB_k \gets {\rm min}\{ H_0(y^{k,t}), UB_k \}, y^k \gets y^{k,t} $
            \EndIf
            \If{ $ UB_k \leq LB(Y^{k,t}) $ }
                \State $ F \gets F \cup \{ Y^{k,t} \} $
            \EndIf
            \State $ \Pi_k = \{Y|Y \in \Pi_{k-1} \cup \{ Y^{k,t}\}, Y \notin F\} $
            \State $ LB_k \gets {\rm min}\{ LB(Y) | Y \in \Pi_k \} $
        \EndFor
        \If{$ UB_k - LB_k \leq \epsilon $}
            \State break
        \EndIf
    \EndFor
\EndIf
\State $ y \gets $ branch-and-bound searching for integer solution with $ y^k $
\State \Return $ y $
\end{algorithmic}
\end{algorithm}

\section{Experiments and Analysis}\label{sec:experiments}

This section details the experiments conducted with respect to user count, user requirement category count, server average computing resources, and different system scales to investigate the algorithms' performance in different situations. In addition, to further study the algorithms' performance, this section also presents the algorithms' computing complexity and appropriate system size for our algorithms, which severely impacts the performance. Furthermore, suggestions for speeding up the algorithms are described to help the algorithm be better applied to real environments. Lastly, we present some experiments on the minimum deployment to study whether it affects the algorithms' performance or not.

\subsection{Experimental Setup}


All the services and servers are randomly generated in all experiments. The input/output data size ranged from 0--2000 KB. Service abilities ranged from 100--400, and computing resources ranged 1--3 units, where the computing resources of services with minimum computing resources were treated as 1 unit computing resources. The dependency graph $ DG $ was also generated randomly in each experiment, whose function chain length ranged 1--7 without cycles. First, several functions are selected uniformly at random as the functions requested by the users. After that, the function chain of each function is generated. During the chain generation, the patterns in previously generated chains are followed when generating new chains to avoid cycles. For example, once the chain $ A \rightarrow B \rightarrow C $ is generated, if $ B $ is selected as the successor of $ A $ in another chain, then $ C $ is automatically chosen as the successor of $ B $. \figurename~\ref{fig:sdg_example} shows the service dependency graph generated randomly in one experiment as an example. A circle with -1 stands for the users, and circles with other numbers denote the different services. The arrows between circles represent the respective call dependency.

\begin{figure}[!t]
    \centering
    \includegraphics[width=\linewidth]{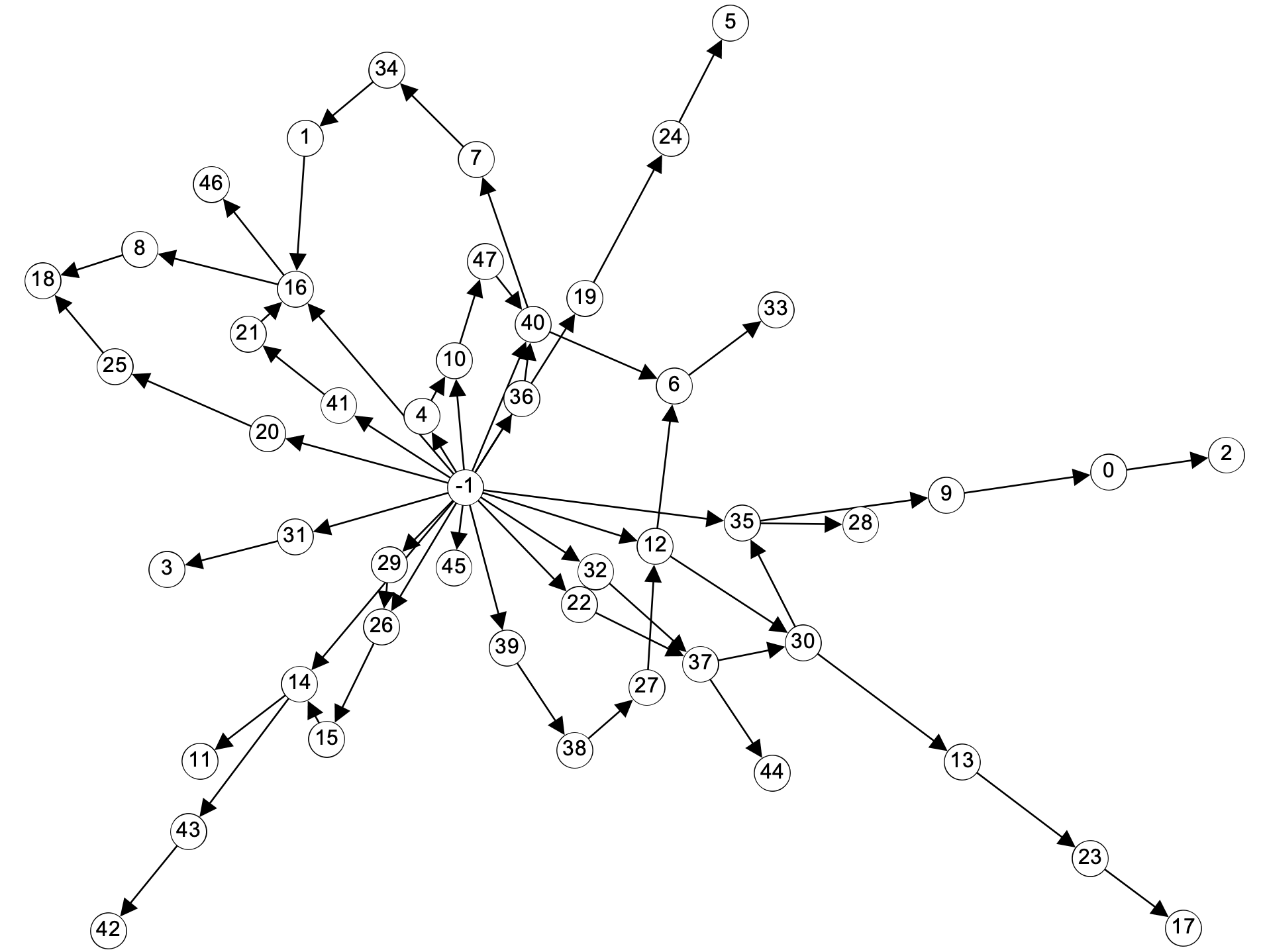}
    \caption{Service dependency graph generated randomly with 48 services}
    \label{fig:sdg_example}
\end{figure}

The servers' computing resources ranged between different experiments, and the delay and bandwidth ranged 1--10ms and 50--1000 MB/s, respectively. In addition, the user requirements were assigned to each server randomly, and the user count varied across different experiments.

\begin{figure*}[!t]
   \centering
   \subfloat[Average response time w.r.t. user count in Experiment~1.1]{
    \includegraphics[width=0.46\linewidth]{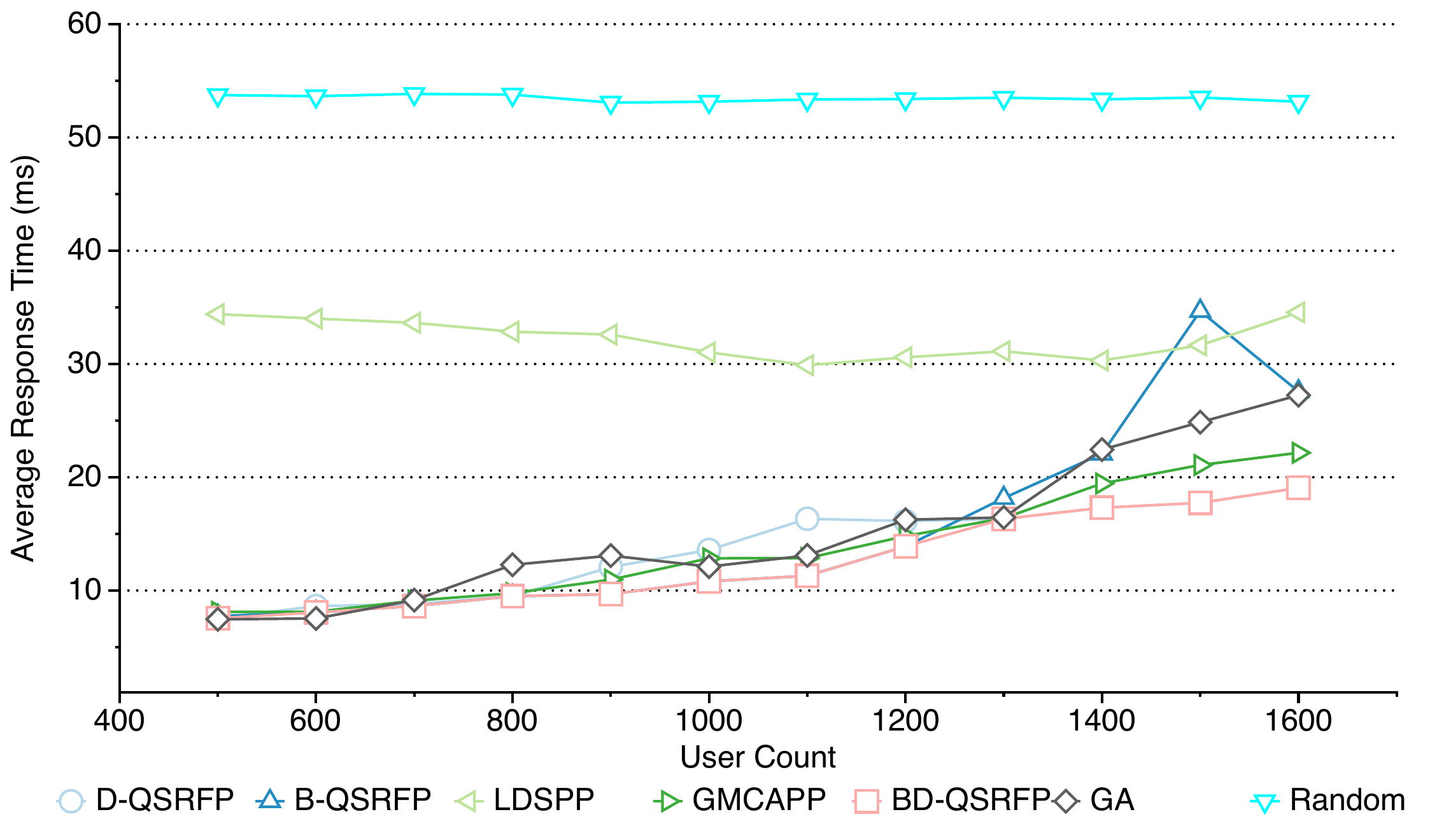}\label{fig:experiment_1_a5}
   }
   \subfloat[Execution time w.r.t. user count in Experiment~1.1]{
    \includegraphics[width=0.46\linewidth]{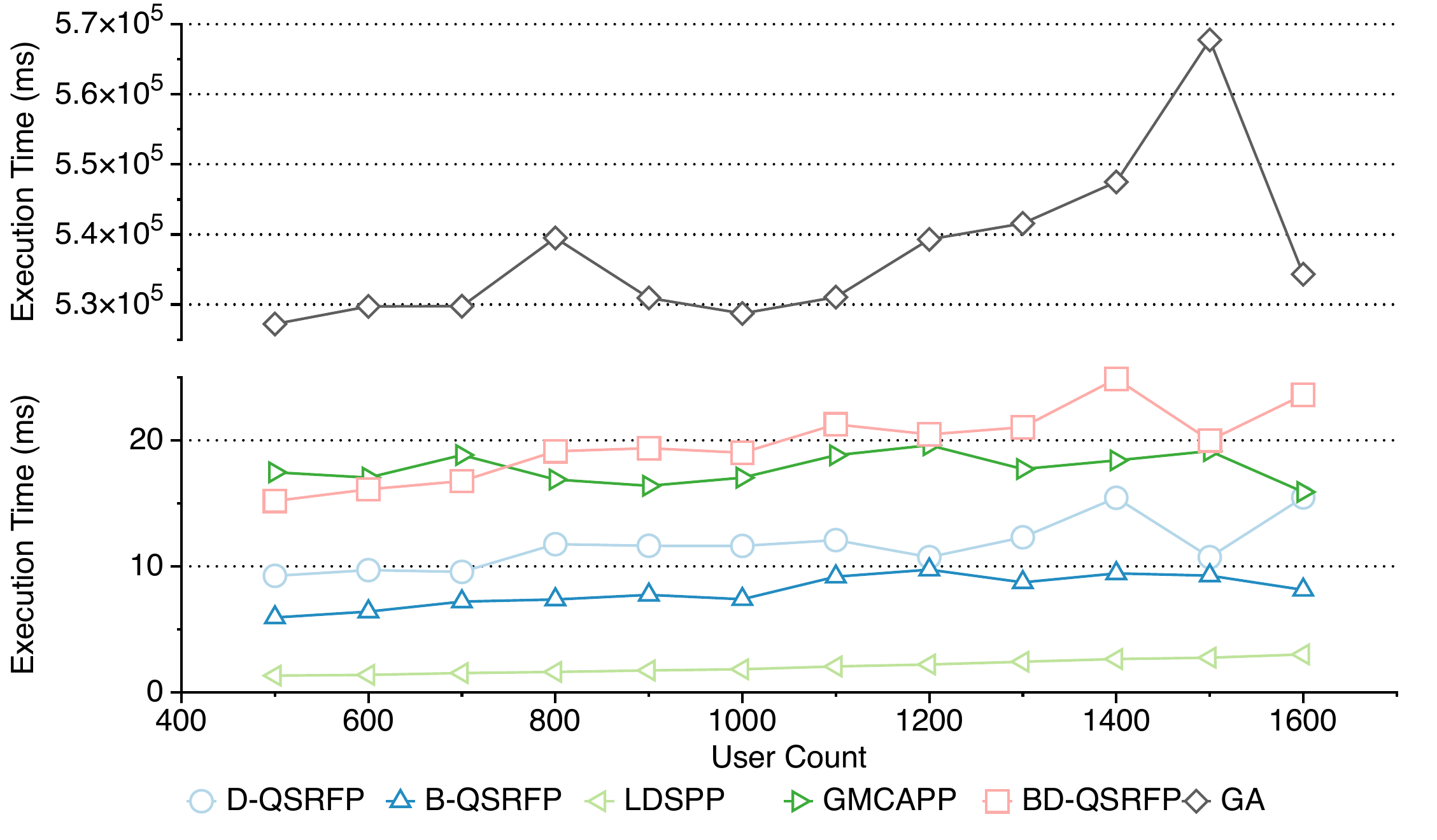}\label{fig:experiment_1_t5}
   }
   \\
   \subfloat[Average response time with w.r.t. user count in Experiment~1.2]{
    \includegraphics[width=0.46\linewidth]{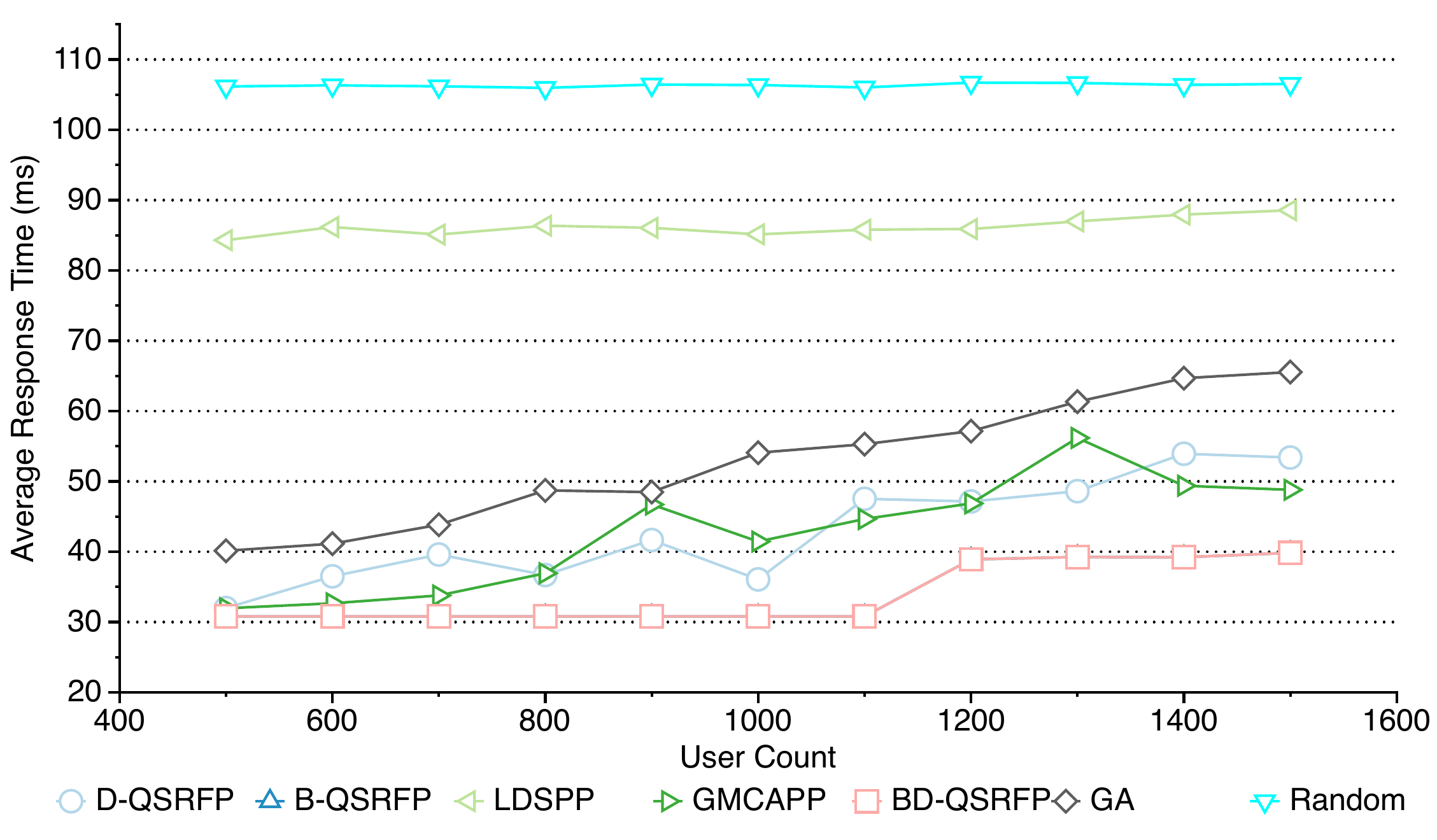}\label{fig:experiment_1_a10}
   }
   \subfloat[Execution time with w.r.t. user count in Experiment~1.2]{
    \includegraphics[width=0.46\linewidth]{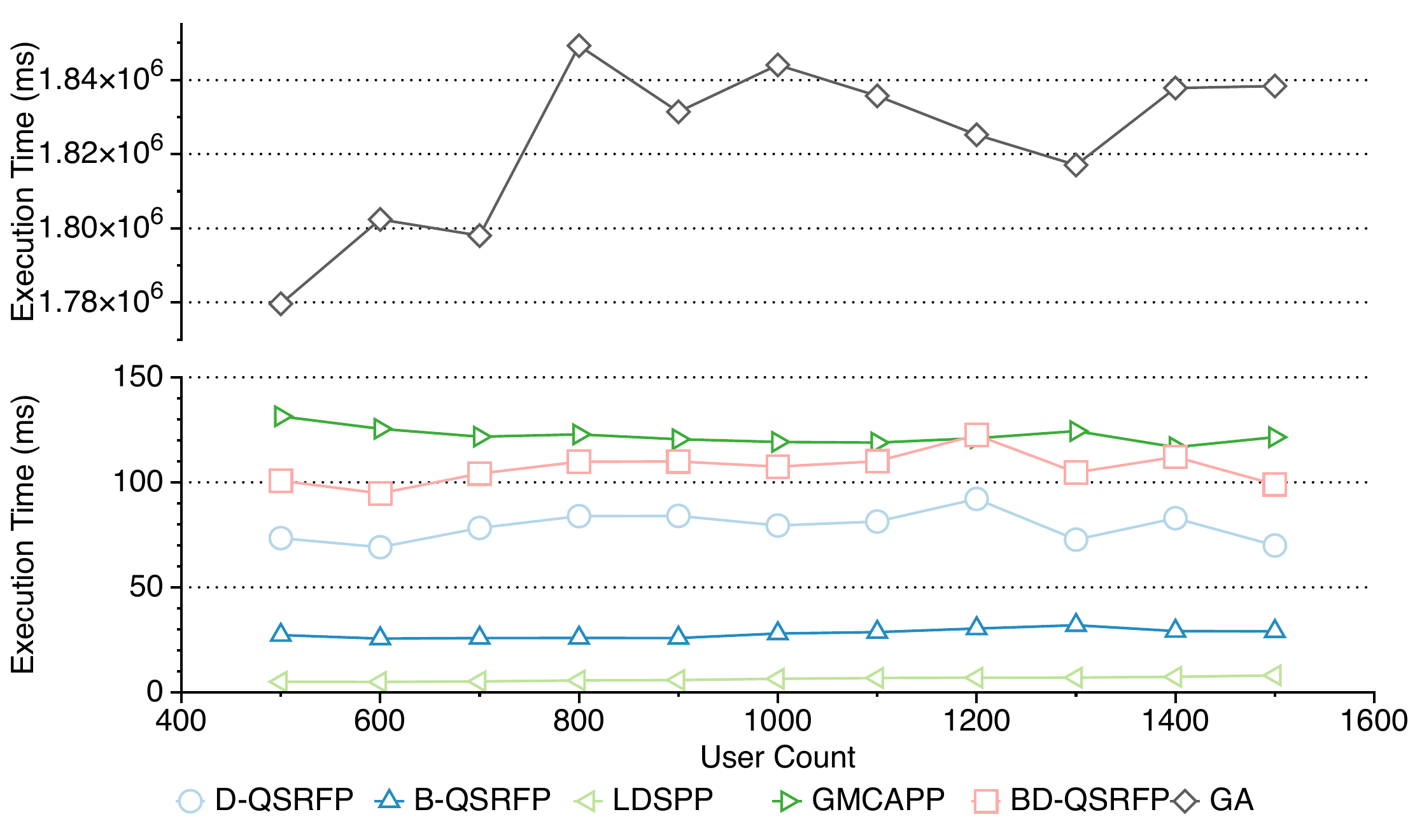}\label{fig:experiment_1_t10}
   }
   \caption{Experiment results w.r.t. user count in Experiment~1}\label{fig:experiment_1}
\end{figure*}

There were three kinds of our algorithms: BFS based algorithm B-QSRFP, DFS based algorithm D-QSRFP, and the combination of B-QSRFP and D-QSRFP denoted as BD-QSRFP. BD-QSRFP runs both B-QSRFP and D-QSRFP separately, and the best result is returned. 

The following algorithms were chosen to be compared with our proposed algorithms:

\begin{itemize}
    \item LDSSP: A lightweight decentralized service placement policy for multi-component application in fog computing proposed in~\cite{LDSPP} which places most popular services as close to the users as possible for lower latency and better network usage.
    \item GMCAPP: An efficient multi-component application online placement algorithm in MEC introduced in~\cite{GMCAPP} which aims to minimize the total cost of running the applications. It was modified with Theorem~\ref{theorem} as the optimization goal, and the deployment cost was changed to be updated after deploying an instance.
    \item GA: Genetic Algorithm is a widely used evolutionary algorithm~\cite{whitley1994genetic} that can solve the placement problem. The main operators of a GA are selection, crossover, and mutation. A feasible deployment scheme $ X $ was used as an individual in our experiments, and the two-point crossover was applied. Randomly delete/add/move an instance were implemented as the mutation with the same probability. 400 population size and a maximum of 400 rounds were used based on preliminary experiments. The tournament selection was adopted for crossover, and mutation probability was set to 0.3. We picked these parameters considering the performance and execution time: higher values lead to long execution times, and lower values make the algorithm hard to converge even on small-size data.
    \item Random: Randomly generating a deployment scheme. It runs for 100 times, and the average result was used as the final result to establish a baseline.
\end{itemize}

It should be noted that due to the high computing complexity of FPP, the problem was converted as QSRFP in all algorithms, which means Equation~\ref{eqa:T_i_j} was used instead.

All the algorithms were implemented in Python 3.6.8 without parallelism and run on the computer with Intel Xeon Gold 5120, 160GB, and CentOS 7. The average response time (ms) and the algorithm execution time (ms) were picked as the metrics. The Random algorithm is for performance comparison only, and its execution time is not included in all experiments. We conducted four experiments to investigate our algorithms' performance under different situations: user count, user requirement category count, server average computing resources, and different system scales.

\subsection{Experiment~1: User Count}

User count is one factor that could affect the deployment scheme due to the service capability of each service. When the number of users increases, the system needs to deploy more instances to satisfy the increasing request frequency. After taking the service dependencies into consideration, the system is also required to deploy more instances for services called by other services.

In Experiment~1, two experiments with different numbers of servers and services were conducted. Experiment~1.1 had 5 servers and 23 services, and Experiment~1.2 had 10 servers and 50 services. The user count ranges between 500--1600 and 500--1500 in Experiments~1.1 and 1.2, respectively. There is no 1600 user count in Experiment~1.2 because the server's computing resource is insufficient even deploying each service with the minimum number.

The results are shown in \figurename~\ref{fig:experiment_1}. For Experiment~1.1, it could be seen that D-QSRFP and B-QSRFP had similar performance overall compared to GMCAPP and GA in \figurename~\ref{fig:experiment_1_a5}, while D-QSRFP and B-QSRFP took less execution time than GMCAPP. BD-QSRFP got better performance than other algorithms but with more execution time than other algorithms except for GA in \figurename~\ref{fig:experiment_1_t5}. It is worth pointing out that GA had much more execution time than other algorithms because it spends lots of time calculating the fitness of each individual with Equation~\ref{eqa:TX_qsrfp}. Because GA is the evolutionary algorithm, it must calculate the fitness after generating new individuals after crossover and mutation. Moreover, GA also needs to repair the infeasible individuals after crossover and mutation because the individuals after crossover and mutation may violate constraints, which also incurs execution cost.

In Experiment~1.2, D-QSRFP performed similar to GMCAPP with less execution time as shown in \figurename~\ref{fig:experiment_1_a10} and~\ref{fig:experiment_1_t10}. B-QSRFP and BD-QSRFP always performed best, and the execution time of BD-QSRFP was less than GMCAPP. The GA results were worse than GMCAPP, D-QSRFP, and B-QSRFP, and the execution time of GA also increased a lot compared to Experiment~1.1, which was caused by more services and servers in Experiment~1.2. With more services and servers, GA had to spend more time on fitness evaluation, and 400 rounds were not enough to find better solutions than other algorithms.

For LDSPP and Random algorithms, their performance was stable all the time in both Experiments~1.1 and~1.2. Though LDSPP had the lowest execution time, it performed worse than others except for Random. The results of Random were stable in the whole experiment due to the same service set. Even though the instance count of each service differs between different user counts, the Random algorithm always tries to evenly deploy the instances of each service across each server. As a result, the evaluated value of Equation~\ref{eqa:TX_qsrfp} is a constant, and the Random algorithm gets similar results. It should be noted that the execution time of all algorithms except for GA increased slightly with user count because more instances needed to be deployed.

From Experiment~1, it could be seen that GA is not suitable as an online algorithm for this problem because it gets harder to find a better deployment scheme within acceptable execution time than other algorithms. B-QSRFP, D-QSRFP, and GMCAPP can get similar performance, while B-QSRFP and Q-QSRFP need less execution time. In order to further study the algorithms' performance, experiments w.r.t. user requirement category count was conducted.

\begin{figure*}[!t]
   \centering
   \subfloat[Average response time w.r.t. user requirement category count]{
    \includegraphics[width=0.46\linewidth]{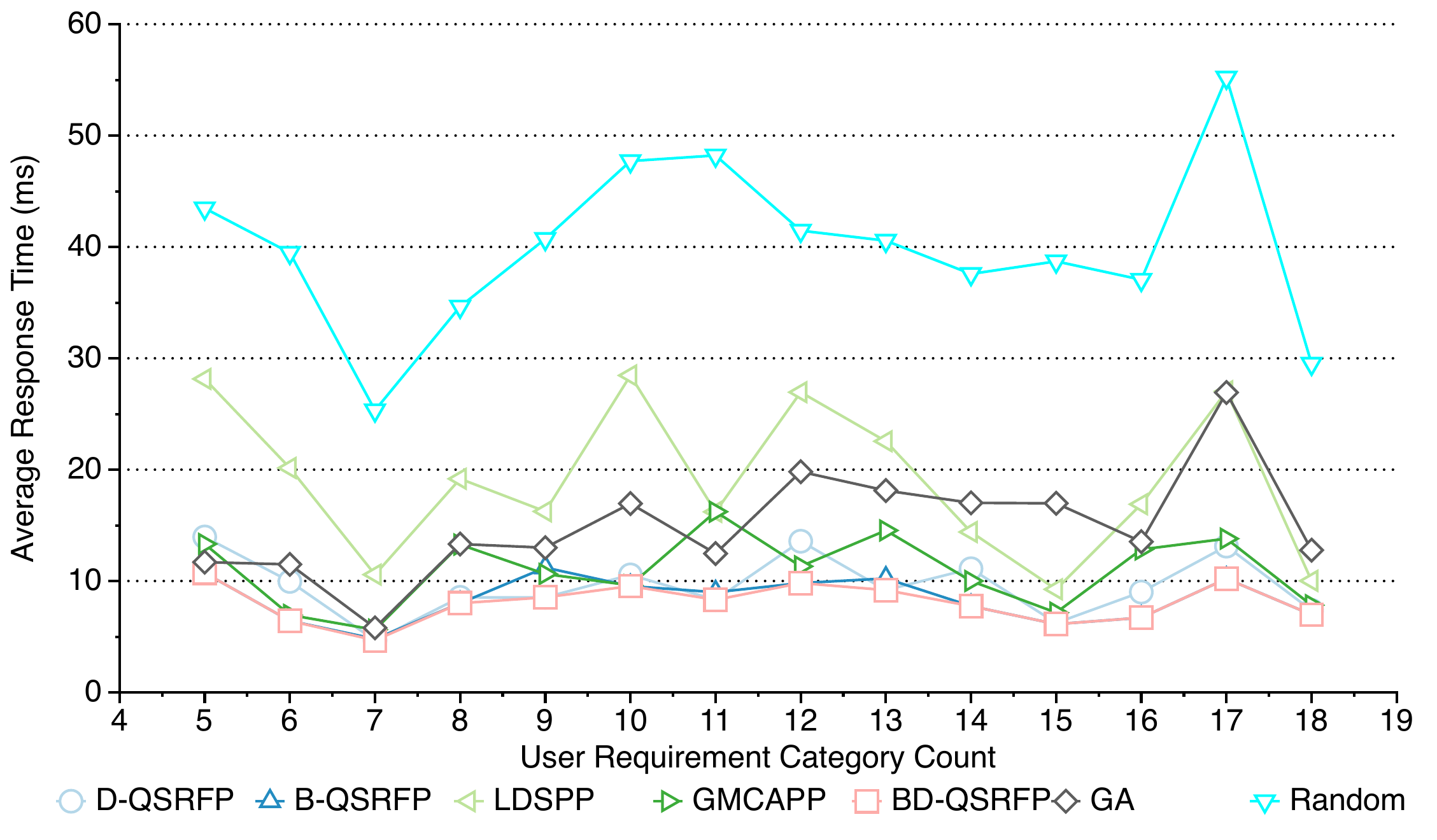}\label{fig:experiment_2_a}
   }
   \subfloat[Execution time w.r.t. user requirement category count]{
    \includegraphics[width=0.46\linewidth]{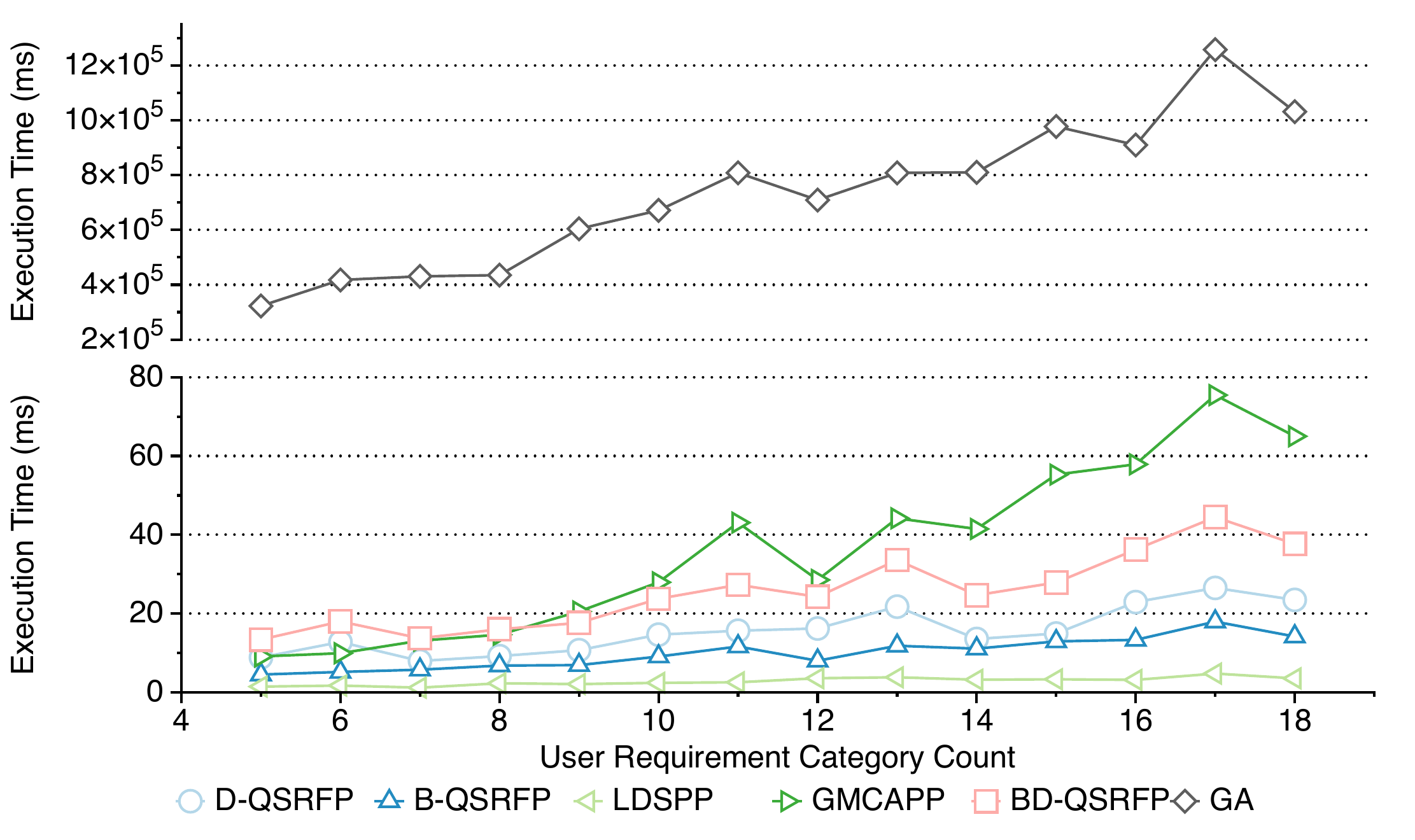}\label{fig:experiment_2_t}
   }
   \caption{Experiment results w.r.t. user requirement category count in Experiment~2}\label{fig:experiment_2}
\end{figure*}

\begin{figure*}[!t]
   \centering
   \subfloat[Average response time w.r.t. average server resources]{
    \includegraphics[width=0.46\linewidth]{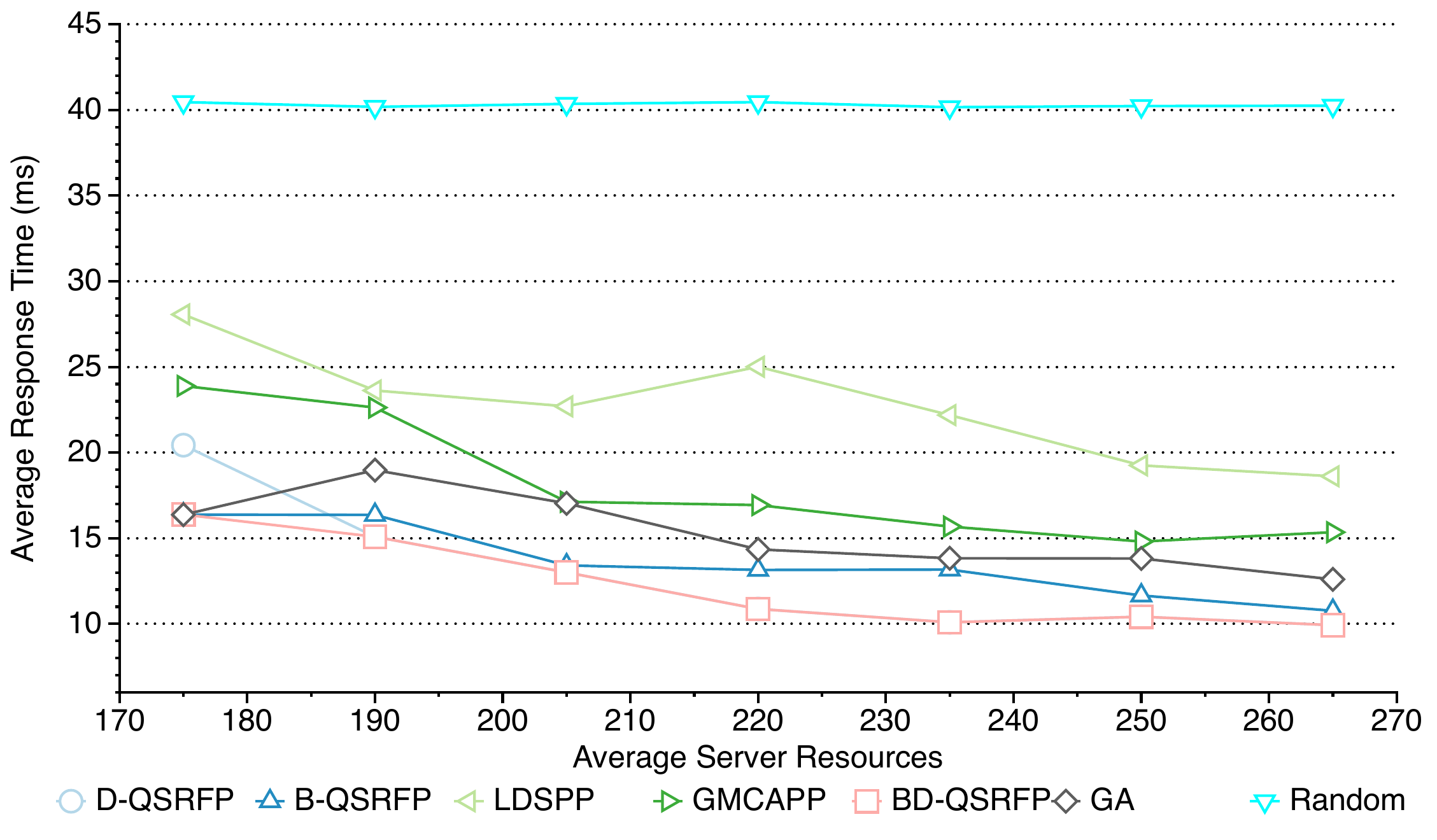}\label{fig:experiment_3_a}
   }
   \subfloat[Execution time w.r.t. average server resources]{
    \includegraphics[width=0.46\linewidth]{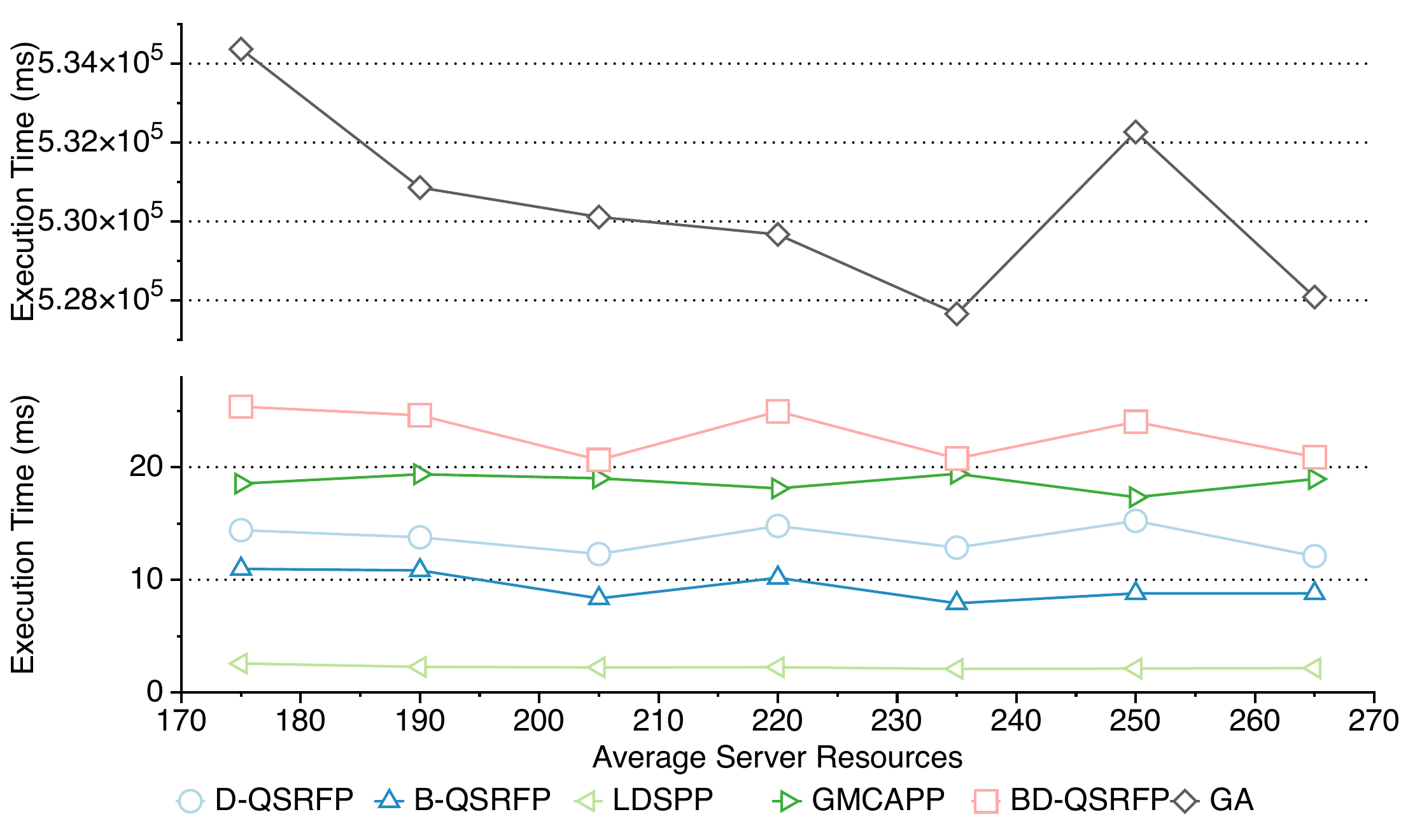}\label{fig:experiment_3_t}
   }
   \caption{Experiment results w.r.t. average server resources in Experiment~3}\label{fig:experiment_3}
\end{figure*}

\subsection{Experiment~2: User Requirement Category Count}

The user requirement category count shows how many different functions are called by users directly, and each has a function chain with an average length of 4. The user requirement category count varies from 5 to 18, and the service dependency also changes. The number of services increases with the user requirement category count and ranges between 15--48, while the number of servers and users keeps the same all the time.

The experiment results are shown in \figurename~\ref{fig:experiment_2}. We can see that the performance of B-QSRFP and D-QSRFP were better than other algorithms with 7, 8, 10, 11, 13, and 15--18 user requirement category count. In other cases, one of B-QSRFP and D-QSRFP got worse results than others, but the other could get the best performance, and as a result, BD-QSRFP was always better than everybody else. GA could get similar results as BD-QSRFP initially, and the results became worse as the number of services increases; this might be due to the increase in dimensionality and an insufficient computational budget.

In \figurename~\ref{fig:experiment_2_t}, all the algorithms got higher execution time with larger count of user requirement category. Less execution time was needed for D-QSRFP and B-QSRFP compared with other algorithms except for LDSPP. The execution time of G-MCAPP was lower than BQ-QSRFP at the beginning and higher than BQ-QSRFP as the number of services increases. Random got different results as the number of services change due to the different service set in this experiment.

The results of Experiment~2 show that B-QSRFP and D-QSRFP can outperform the other approaches in both time in quality, and BD-QSRFP is the best while requiring the highest execution time. Considering the servers' computing resources keep the same in this experiment, we conducted more experiments w.r.t. server computing resources to study the algorithms' performance under different system computing resource usages.

\subsection{Experiment~3: Average Server Resources}

\begin{figure*}[!t]
   \centering
   \subfloat[Average response time w.r.t. system scale in Experiment~4.1]{
    \includegraphics[width=0.46\linewidth]{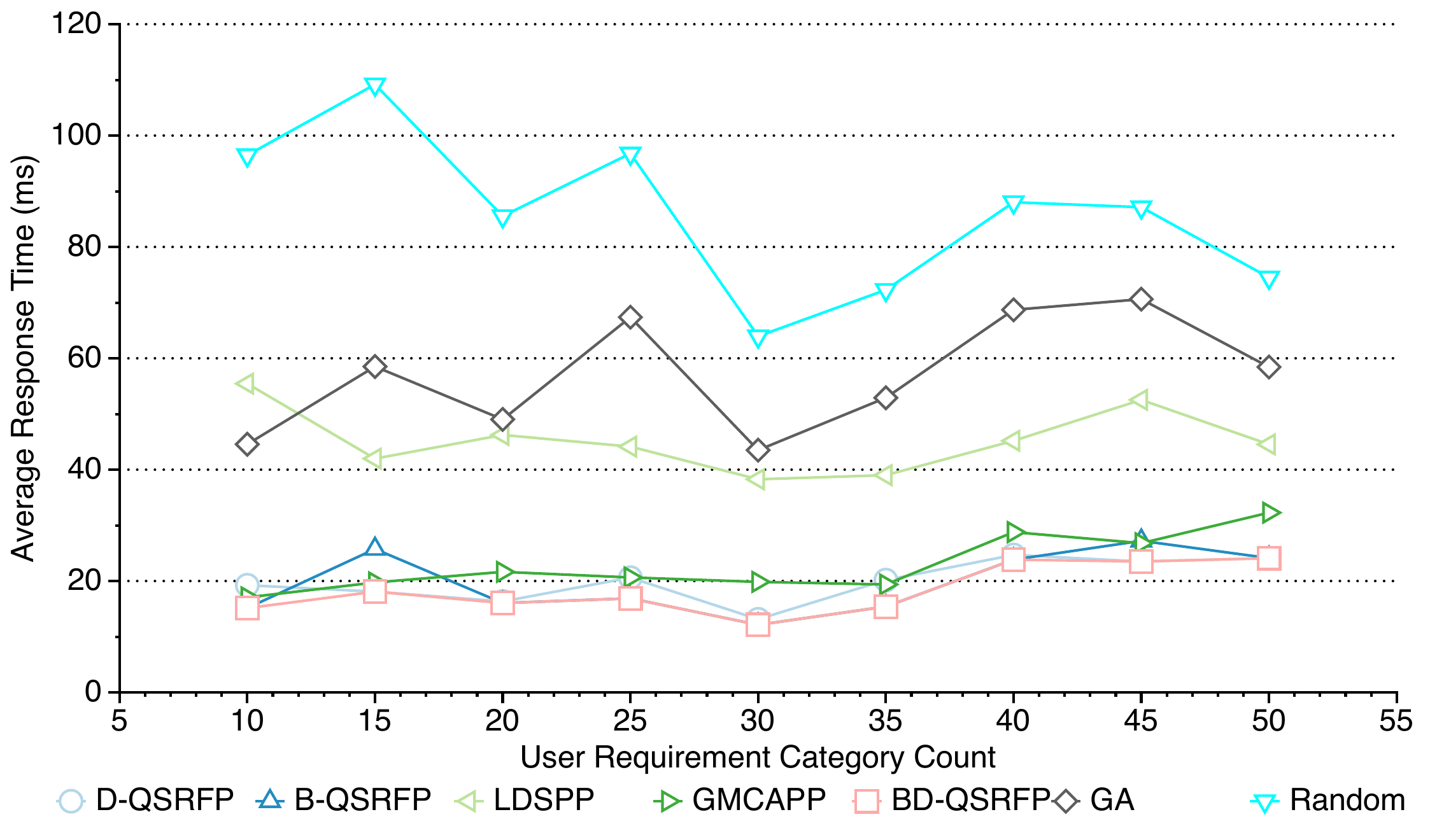}\label{fig:experiment_4_1_a}
   }
   \subfloat[Execution time w.r.t. system scale in Experiment~4.1]{
    \includegraphics[width=0.46\linewidth]{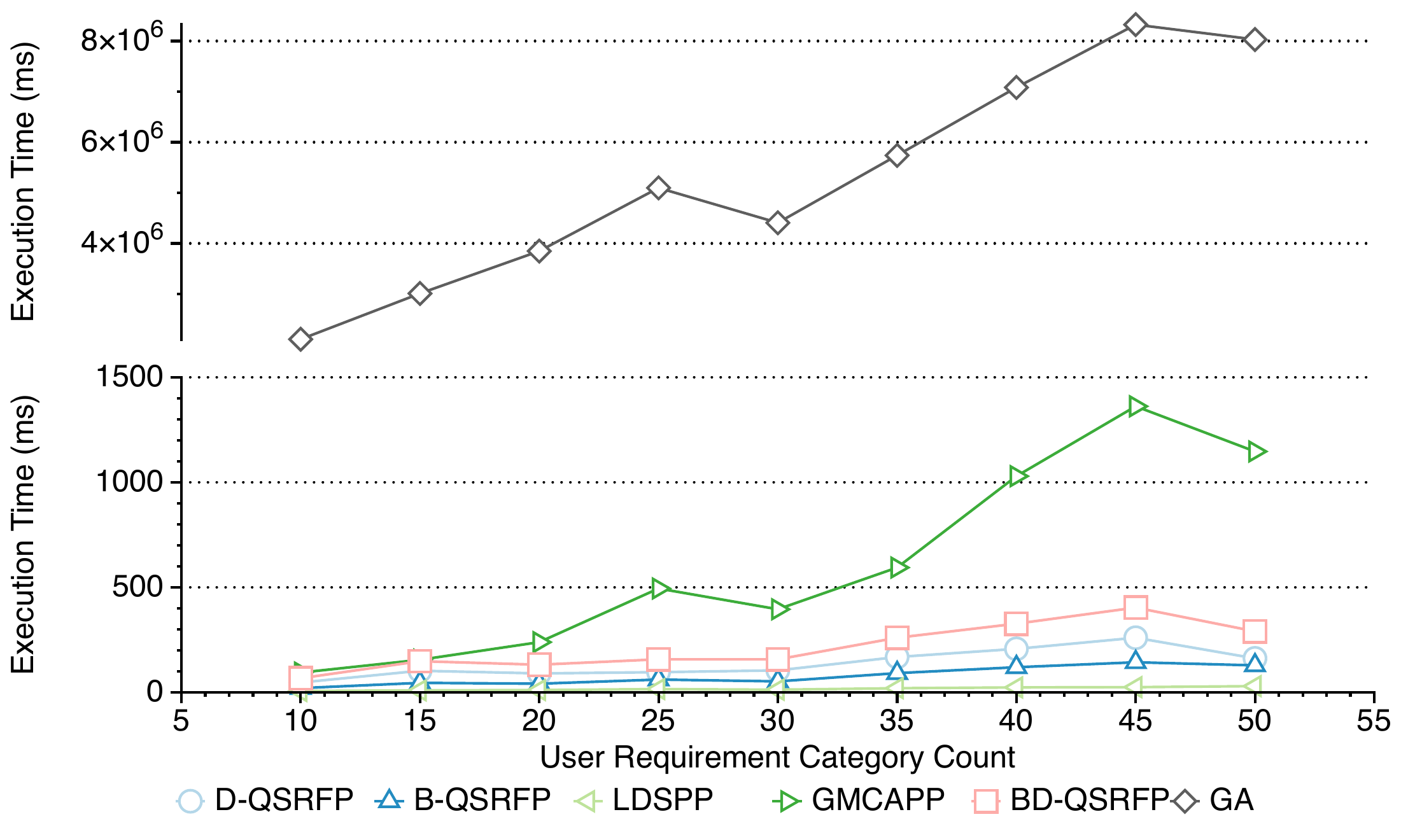}\label{fig:experiment_4_1_t}
   }
   \caption{Experiment results w.r.t. average server resources in Experiment~4.1}\label{fig:experiment_4_1}
\end{figure*}

\begin{figure*}[!t]
   \centering
   \subfloat[Average response time w.r.t. system scale in Experiment~4.2]{
    \includegraphics[width=0.46\linewidth]{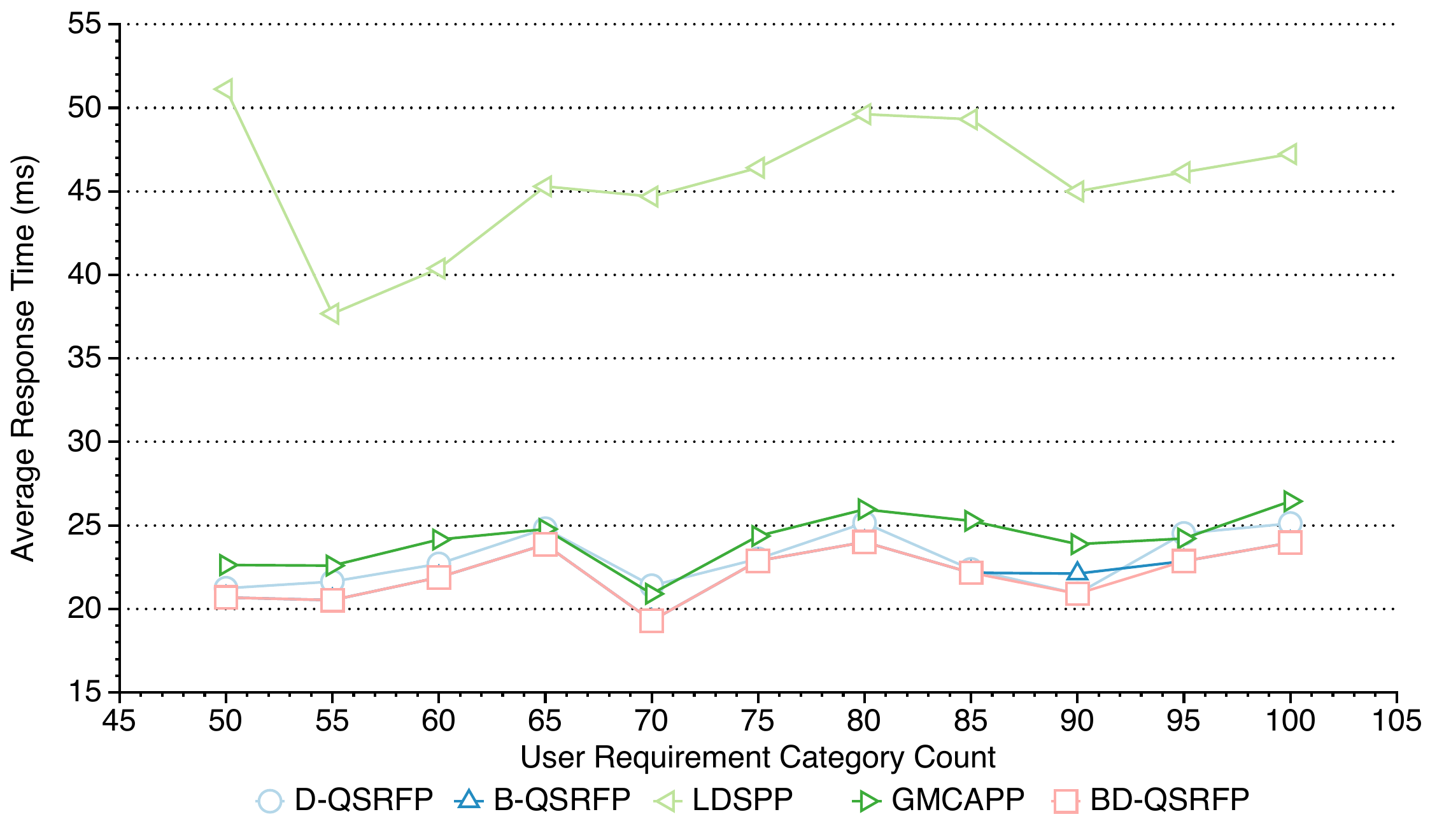}\label{fig:experiment_4_2_a}
   }
   \subfloat[Execution time w.r.t. system scale in Experiment~4.2]{
    \includegraphics[width=0.46\linewidth]{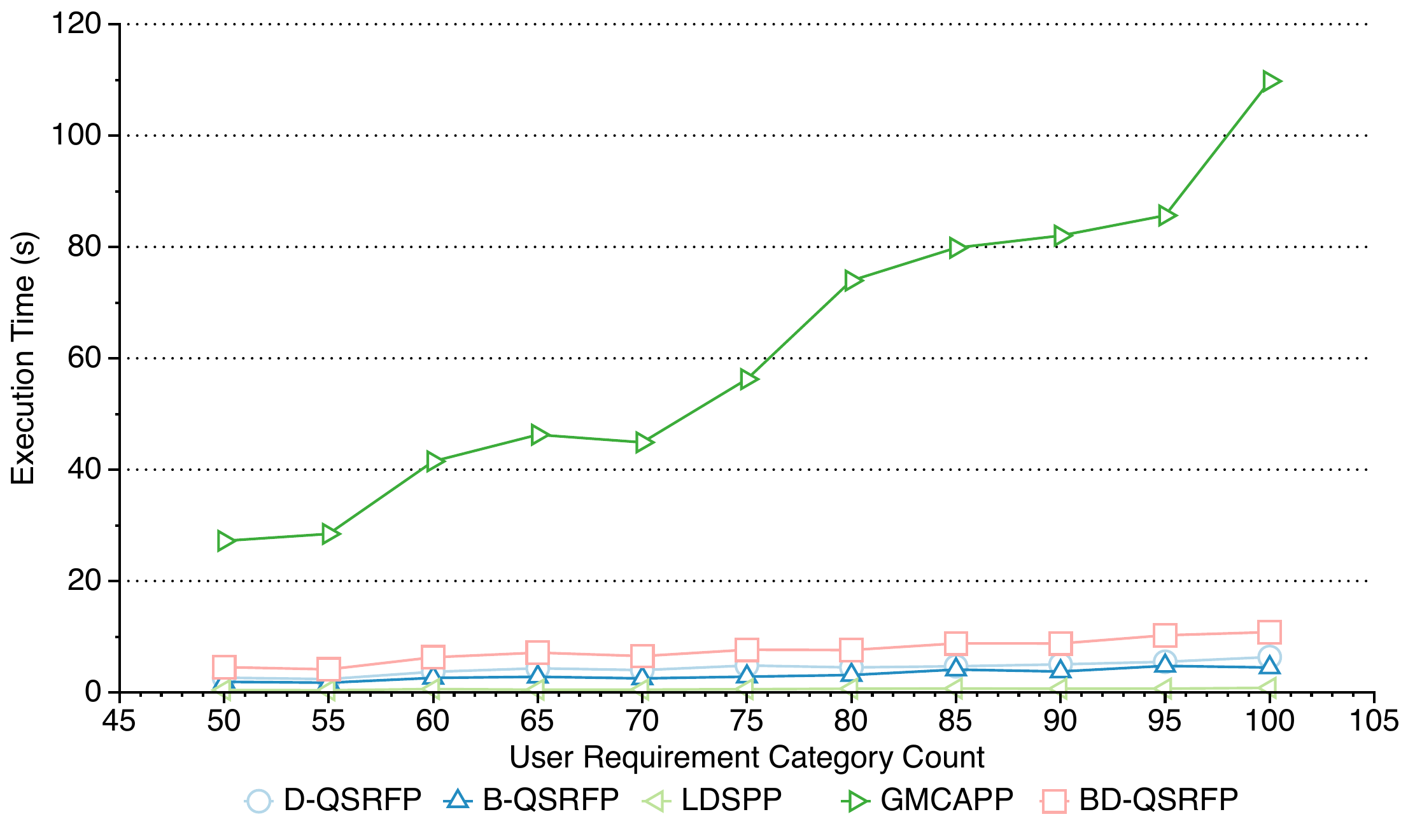}\label{fig:experiment_4_2_t}
   }
   \\
   \subfloat[Average response time with w.r.t. system scale in Experiment~4.3]{
    \includegraphics[width=0.46\linewidth]{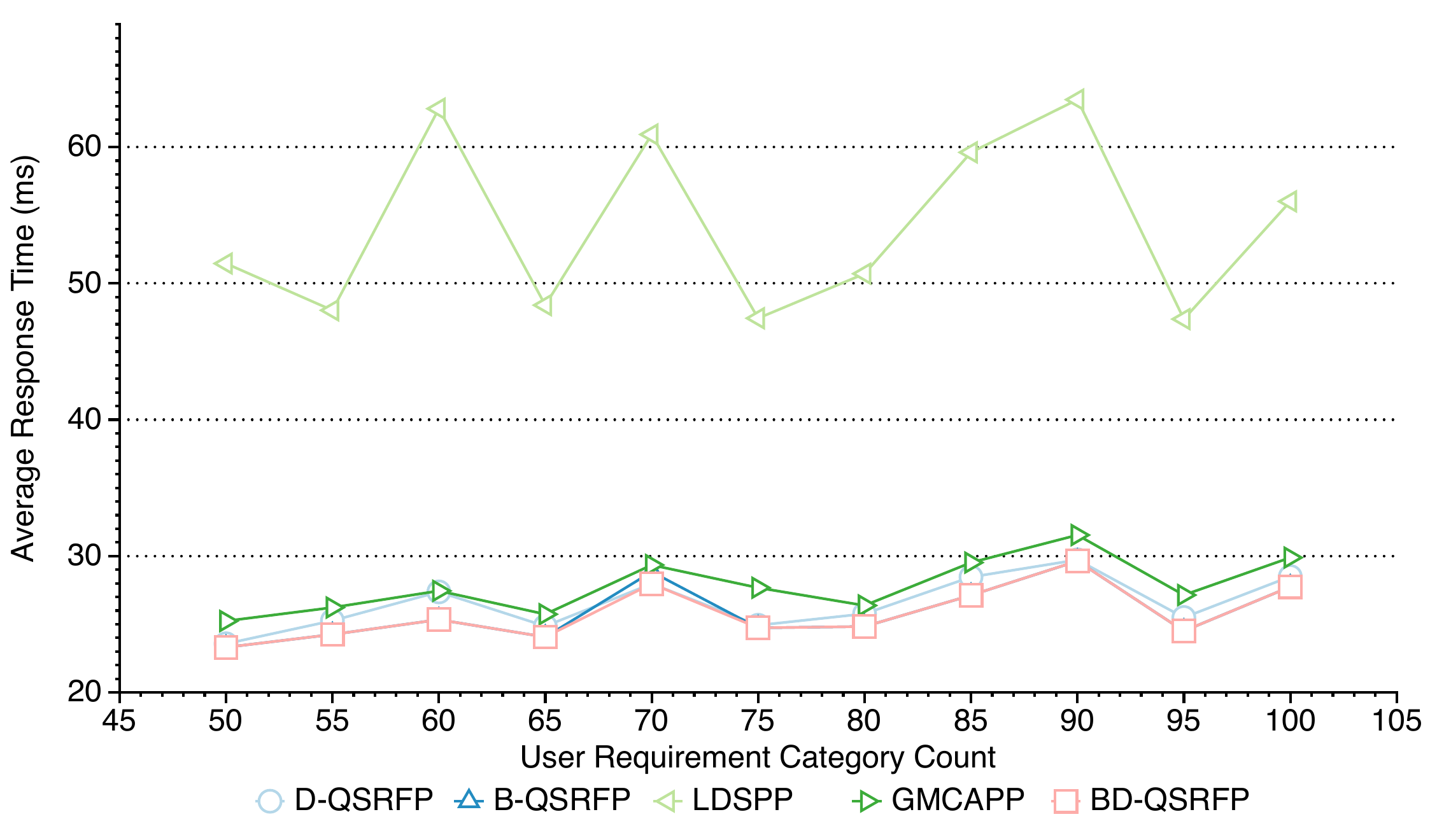}\label{fig:experiment_4_3_a}
   }
   \subfloat[Execution time with w.r.t. system scale in Experiment~4.3]{
    \includegraphics[width=0.46\linewidth]{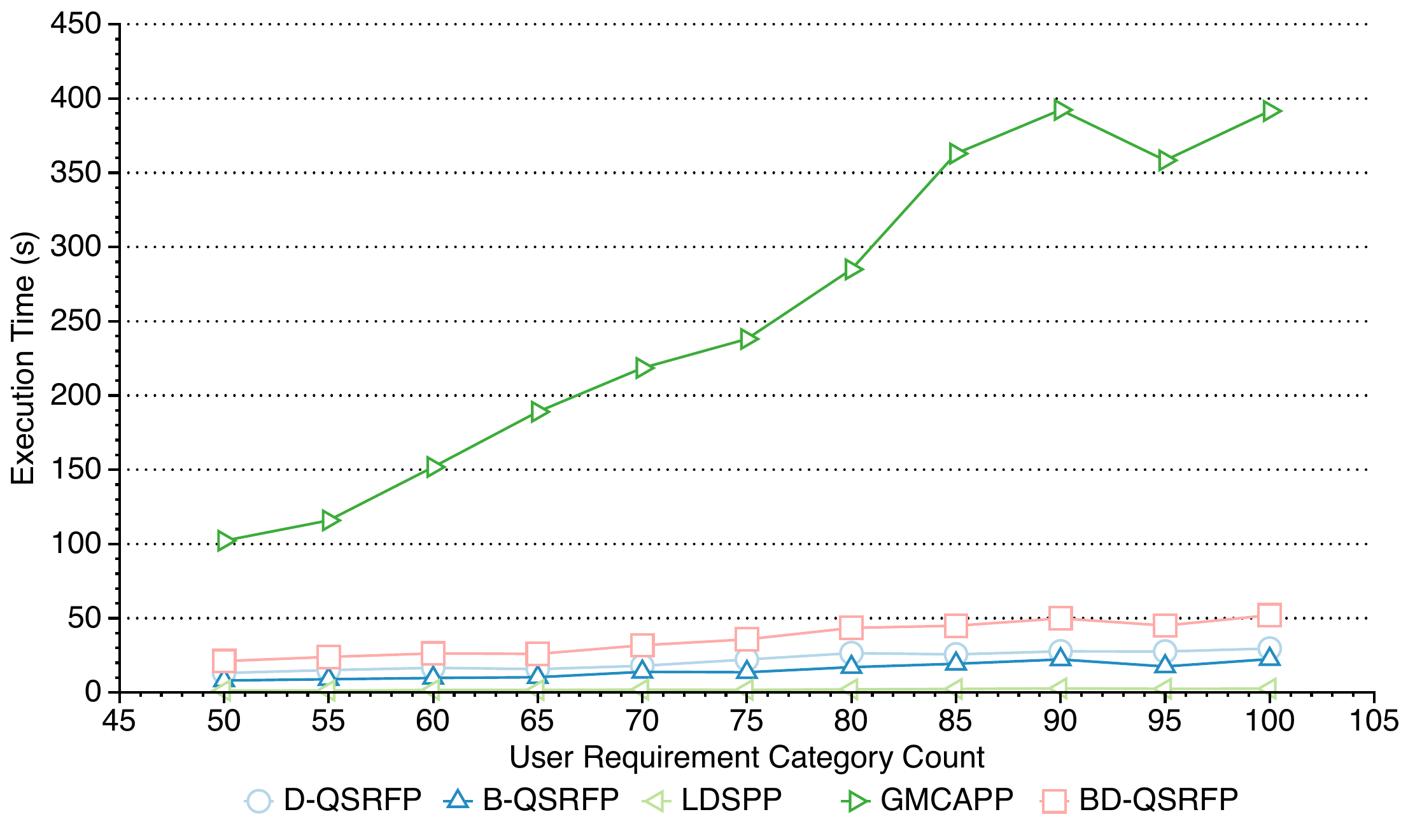}\label{fig:experiment_4_3_t}
   }
   \caption{Experiment results w.r.t. system scale in Experiment~4.2 and 4.3}\label{fig:experiment_4_23}
\end{figure*}

In this experiment, the average server resources are represented in the number of units, which stands for the computing resources used by the service with the minimal resource requirement. The number of servers keeps the same, and the average server resources range from 175 units to 265 units. The user count and the user requirement category count keep the same, and 175 units are the minimal average server resources that can hold enough instances to meet user requirements. As the average server resources increase, the system can deploy more service instances on each server.

Based on the results shown in in \figurename~\ref{fig:experiment_3}, B-QSRFP and D-QSRFP outperformed other algorithms in most situations in \figurename~\ref{fig:experiment_3_a}, and their execution time was lower than others except for LDSPP, while LDSPP got worse performance than them. Due to the small size of servers and services, GA got better results than GMCAPP in this experiment with the most execution time, which is still unacceptable as an online algorithm. BD-QSRFP still had the best performance in the whole experiment, and the performance of Random nearly kept the same as explained in Experiment~1.

It should be noted that the execution time of all algorithms except for GA is generally stable, but there are still some differences as the average server resources change in \figurename~\ref{fig:experiment_3_t}. It is because the instance number of each service keeps the same due to the same service set and servers. But with different average server resources, the instance number that can be run on each server differs, and the algorithm needs to spend more time to find an adequate server for each instance again when the resources of the server used before are insufficient. Thus the execution time varies.

The results in this experiment show that our algorithm B-QSRFP, D-QSRFP, and BD-QSRFP outperform others with better average response time and lower execution time. To investigate their performance with more servers and services, we conducted more experiments with different system scale.

\subsection{Experiment~4: System Scale} \label{subsec:experiment_4}

In this experiment, we enlarged the number of servers and services to check the algorithms' performance. Experiment~4.1 was conducted with 10 servers and 10-50 user requirement category count, and the maximal service size was 130, while the user count kept the same. The results are shown in \figurename~\ref{fig:experiment_4_1}.

The results in \figurename~\ref{fig:experiment_4_1_a} and~\ref{fig:experiment_4_1_t} are similar to our previous experiments. B-QSRFP and D-QSRFP could get the best or similar results compared to others with less execution time, and BD-QSRFP had the best performance all the time, while its execution time was higher. GA performed worse than other experiments due to the larger size of services and servers with more execution time. GMCAPP's execution time increased a lot when the user requirement category count increased.

To further study the algorithms' performance, Experiments~4.2 and~4.3 were conducted. There were 50 servers and 100 servers in Experiments~4.2 and~4.3, respectively. The user requirement category count ranged between 50--100 in both experiments, and the service size ranged between 150-320. Because GA performed worse than others, GA and Random were not used in Experiments~4.2 and~4.3. The results are shown in \figurename~\ref{fig:experiment_4_23}.

Based on the results in \figurename~\ref{fig:experiment_4_2_a} and~\ref{fig:experiment_4_3_a}, B-QSRFP and D-QSRFP also got better results than GMCAPP and LDSPP, and BD-QSRFP had the best performance even with large scale system. The execution time of B-QSRFP, D-QSRFP, and BD-QSRFP were lower than GMCAPP in \figurename \ref{fig:experiment_4_2_t} and~\ref{fig:experiment_4_3_t}. It should be pointed out that the maximum execution times of BD-QSRFP were about 11s and 52s in Experiments~4.2 and~4.3, respectively, which is too costly for an online algorithm. However, 100 servers are too many for the microservice system, and other technologies could improve the algorithms' speed. For more details, please refer to Section~\ref{subsubsec:system_size} and~\ref{subsubsec:speed_up}.

\subsection{Analysis}

In this section, we analyse the algorithms' computing complexity and discuss appropriate system sizes. We suggest ways for speeding up the algorithms and investigate the performance impact of a minimum instance deployment.

\subsubsection{Computing Complexity Analysis} \label{subsubsec:complexity}

The complexity of Algorithm~\ref{algo:best_server} is $ O(2|N|^3) $. In Algorithm~\ref{algo:deploy}, it would be $ \Omega(2|N|^3) $ in the ideal situation, which means sufficient computing resources and no dependency at all. In the worst situation, it would not exceed $ O(2 + 2|S|)|N|^4 $: all the services are needed to be re-deployed, and Algorithm~\ref{algo:best_server} must run $ |N| $ times to deploy all instances. For B-QSRFP and D-QSRFP, it should be $ O(2 + 2|S|)|S||N|^4 $.

\subsubsection{Appropriate System Size} \label{subsubsec:system_size}

According to Section~\ref{subsubsec:complexity}, the server size has an essential impact on our algorithms' performance, which means the microservice system should have an appropriate server size so that the execution time taken by our algorithms is affordable as an online algorithm.

Based on the results in Section~\ref{subsec:experiment_4}, 100+ servers are too many for our algorithms even after adopting the suggestions in Section~\ref{subsubsec:speed_up} to speed up. In our opinion, the microservice system should not hold too many servers at the same time during the service placement or system evolution. On the one hand, the cost of collecting the running information from all servers can be huge, i.e., the bandwidth and the evolution delay due to the data transmission. On the other hand, as a distributed system, the stability of the distributed system should be ensured during the placement or evolution, and centralized algorithms should not be used for too many servers.

Thus, for a large microservice system with thousands of servers, we suggest dividing it into several small groups whose server size does not exceed 100 to reduce the cost during the information collection and keep the distributed system's stability while makes the algorithms' execution time is affordable as online algorithms.

\subsubsection{Suggestions on Speeding up} \label{subsubsec:speed_up}

There are two methods to speed up our algorithms. One is implementing the algorithms in a high-performance programming language like C++. Because our algorithms were implemented in pure Python, the speed should be improved a lot due to the low performance of Python compared to C++.

Another is adopting multi-threads technology. It could be seen that Algorithm~\ref{algo:best_server} plays an important role in the algorithms' performance. The algorithm calculates the average time for each server by assuming deploying one instance on each server, and it is safe to do it with $ |N| $ threads, which also can speed up the algorithms.

Using B-QSRFP or D-QSRFP instead of BD-QSRFP in a large-scale system also helps to reduce execution time. Though B-QSRFP and D-QSRFP cannot get the best results, their execution time is lower than BD-QSRFP.

\subsubsection{Impact of Minimum Deployment on Performance}

The strategy used in our algorithms only deploy instances that just satisfy the total request rate from users and services. For example, assuming the total request rate of $ s_i $ is $ \lambda_i $, then $ \lceil \frac{\lambda_i}{\mu_i} \rceil $ instances are deployed in the system. According to Equation~\ref{eqa:T_i_j}, it is possible to improve the average response time by deploying more instances. To investigate the performance improvement with more instances, some experiments were conducted.

At the end of Algorithm~\ref{algo:bfs_placement} and~\ref{algo:dfs_placement}, the algorithms deploy more instances in order to lower average response time. They search the service $ s_i $ and the server $ n_j $ that can make the average response time decrease the most when deploying an instance of $ s_i $ on $ n_j $, and deploy $ s_i $ on $ n_j $ until  the constraints are satisfied. The dataset is the same as Experiment~2, and the results are shown in Table~\ref{tab:push_bound}. D-QSRFP and B-QSRFP are abbreviated as D and B, and D' and B' stand for the algorithms that deploy more instances for improvement until meeting constraints.

\begin{table}[htbp]
\caption{Results with/without minimal instances} \label{tab:push_bound}
\centering
\begin{adjustbox}{width=\linewidth}
\begin{tabular}{c|cccc|cccc} 
\hline
\multirow{2}{*}{Count} & \multicolumn{4}{c|}{Average Response Time}                                                           & \multicolumn{4}{c}{Execution Time}  \\ 
\cline{2-9}
                       & D     & D'                                        & B     & B'                                       & D     & D'     & B     & B'         \\ 
\hline
5                      & 13.96 & 13.96                                     & 10.74 & 10.74                                    & 8.87  & 9.34   & 4.53  & 5.22       \\
6                      & 9.97  & 9.97                                      & 6.43  & 6.43                                     & 12.80 & 15.13  & 5.20  & 6.95       \\
7                      & 4.68  & 4.68                                      & 4.80  & 4.80                                     & 7.97  & 9.25   & 5.76  & 7.58       \\
8                      & 8.54  & 8.54                                      & 8.01  & 8.01                                     & 9.21  & 11.31  & 6.81  & 9.00       \\
9                      & 8.54  & 8.54                                      & 11.21 & 11.21                                    & 10.72 & 12.68  & 6.92  & 9.04       \\
10                     & 10.55 & 10.55                                     & 9.55  & {\cellcolor[rgb]{0.753,0.753,0.753}}9.53 & 14.66 & 17.15  & 9.10  & 27.79      \\
11                     & 8.31  & 8.31                                      & 9.00  & 9.00                                     & 15.65 & 18.59  & 11.65 & 14.36      \\
12                     & 13.58 & 13.58                                     & 9.82  & {\cellcolor[rgb]{0.753,0.753,0.753}}9.78 & 16.22 & 18.64  & 8.03  & 164.42     \\
13                     & 9.19  & {\cellcolor[rgb]{0.753,0.753,0.753}}9.17  & 10.24 & 10.24                                    & 21.76 & 61.40  & 11.84 & 15.77      \\
14                     & 11.11 & {\cellcolor[rgb]{0.753,0.753,0.753}}10.71 & 7.75  & 7.75                                     & 13.57 & 332.13 & 11.11 & 14.05      \\
15                     & 6.24  & 6.24                                      & 6.13  & 6.13                                     & 14.94 & 18.00  & 12.94 & 16.18      \\
16                     & 9.01  & 9.01                                      & 6.71  & 6.71                                     & 22.92 & 28.32  & 13.33 & 17.86      \\
17                     & 13.13 & 13.13                                     & 10.22 & 10.22                                    & 26.52 & 30.38  & 17.96 & 22.72      \\
18                     & 7.37  & 7.37                                      & 6.98  & 6.98                                     & 23.47 & 27.23  & 14.18 & 18.02      \\
\hline
\end{tabular}
\end{adjustbox}
\end{table}

The results with gray color show that deploying more instances can improve the average response time at the cost of a much higher execution time. Thus, it should be used only when the execution time is acceptable.

\section{Conclusion and Future Work}\label{sec:conclusion}

This paper has formulated the service placement problem in microservice systems with complex dependencies as a FPP considering the dependencies between services and support for multiple service instances. Due to the high computing complexity of FPP, we have proven that we can convert the FPP to QSRFP, which results in a significant decrease in computational complexity. Based on the QSRFP, we have proposed the greedy-based algorithms B-QSRFP, D-QSRFP, and BD-QSRFP. The experimental evaluation has shown that our algorithms outperform other approaches in both computation time and quality. We have discussed ways of speeding up our approaches and discussed appropriate system sizes.

Possible future work includes the programming framework and system infrastructure supporting self-adaptive microservice system evolution in multiple groups: the large-scale microservice system must be divided into several groups to maintain the stability of the distributed microservice system. The cost of system information collection and evolution time can also be reduced with groups, and each group can evolve and cooperate with other groups.

\appendices


\ifCLASSOPTIONcompsoc
  \section*{Acknowledgments}
\else
  \section*{Acknowledgment}
\fi

Research in this paper is partially supported by the National
Key Research and Development Program of China (No
2018YFB1402500), the National Science Foundation of China
(61832014, 61772155, 61832004), as well as by the Australian Research Council (DP200102364, DP210102670).

\ifCLASSOPTIONcaptionsoff
  \newpage
\fi



%




\bibliographystyle{IEEEtran}
\bibliography{refs.bib}

%

\begin{IEEEbiography}[{\includegraphics[width=1in,height=1.25in,clip,keepaspectratio]{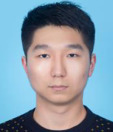}}]{Xiang He} received his B.S. degree from the School of Computer Science and Technology, Harbin Institute of Technology in 2018. He is currently pursuing the Ph.D. degree in software engineering at the Harbin Institute of Technology (HIT), China. His research interests include edge computing, microservice system infrastructure design, and self-adaptive system.
\end{IEEEbiography}

\begin{IEEEbiography}[{\includegraphics[width=1in,height=1.25in,clip,keepaspectratio]{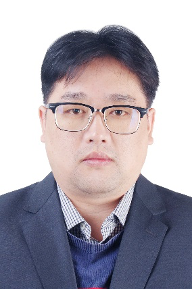}}]{Zhiying Tu} is an associate professor in the Faculty of Computing at Harbin Institute of Technology (HIT). He holds a PhD degree in Computer Integrated Manufacturing (Productique) from the University of Bordeaux. Since 2013, he began to work at HIT. His research interest is Service Computing, Enterprise Interoperability, and Cognitive Computing. He has 20 publications as edited books and proceedings, refereed book chapters, and refereed technical papers in journals and conferences. He is a member of the IEEE Computer Society and of the CCF China.
\end{IEEEbiography}

\begin{IEEEbiography}[{\includegraphics[width=1in,height=1.25in,clip,keepaspectratio]{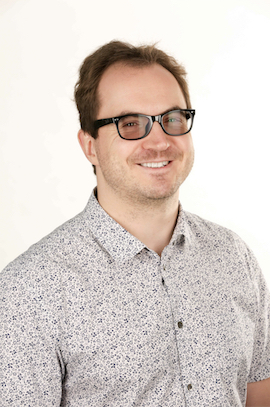}}]{Markus Wagner}
is an associate professor in the School of Computer Science at the University of Adelaide. His works on single- and multi-objective optimisation span theory-driven algorithm development and applications in domains such as search-based software engineering and renewable energy. He has written over 150 articles with over 150 co-authors. He served as founding chair of two IEEE CIS Task Forces, and he has chaired education-related committees within the IEEE CIS.
\end{IEEEbiography}

\begin{IEEEbiography}[{\includegraphics[width=1in,height=1.25in,clip,keepaspectratio]{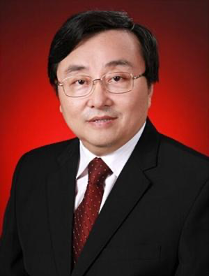}}]{Xiaofei Xu} is a professor at Faculty of Computing, and vice president of the Harbin Institute of Technology. He received the Ph.D. degree in computer science from Harbin Institute of Technology in 1988. His research interests include enterprise intelligent computing, services computing, internet of services, and data mining. He is the associate chair of IFIP TC5 WG5.8, chair of INTEROP-VLab China Pole, fellow of China Computer Federation (CCF), and the vice director of the technical committee of service computing of CCF. He is the author of more than 300 publications. He is member of the IEEE and ACM.
\end{IEEEbiography}


\begin{IEEEbiography}[{\includegraphics[width=1in,height=1.25in,clip,keepaspectratio]{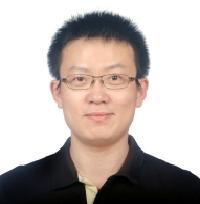}}]{Zhongjie Wang} is a professor at Faculty of Computing, Harbin Institute of Technology (HIT). He received the Ph.D. degree in computer science from Harbin Institute of Technology in 2006. His research interests include services computing, mobile and social networking services, and software architecture. He is the author of more than 80 publications. He is a member of the IEEE.
\end{IEEEbiography}





\end{document}